\documentclass[journal]{IEEEtran}

\usepackage{mathtools}

\usepackage{times}
\usepackage{fixmath}
\usepackage{graphicx}
\usepackage{subcaption}
\usepackage{amsthm,amssymb}
\usepackage[colorlinks,  allcolors=blue]{hyperref}
\usepackage{amsfonts}
\usepackage{multicol}
\usepackage{amsfonts}
\usepackage{cancel}
\usepackage[english]{babel}

\usepackage{algorithm}
\usepackage{algpseudocode}

\makeatletter
\def\BState{\State\hskip-\ALG@thistlm}
\makeatother

\usepackage[labelsep=period]{caption}
\usepackage[figurename=Fig.]{caption}
\usepackage[font=footnotesize]{caption}

\DeclareMathOperator{\sinc}{sinc}

\newtheorem{examplex}{Example}

\newtheorem{lemma}{Lemma}

  {\popQED\endexamplex}
  
\usepackage{mathtools}

\DeclarePairedDelimiter\floor{\lfloor}{\rfloor}

\newcommand\blfootnote[1]{%
  \begingroup
  \renewcommand\thefootnote{}\footnote{#1}%
  \addtocounter{footnote}{-1}%
  \endgroup
}

\begin{document} 

\title{Calculation of the Mean Strain of  Smooth \color{black} Non-uniform Strain Fields Using Conventional FBG Sensors}

\author{Aydin Rajabzadeh,
	Richard Heusdens,
        Richard C. Hendriks,
        and~Roger M. Groves
}





\maketitle 


\begin{abstract}
In the past few decades, fibre Bragg grating (FBG) sensors have gained a lot of attention in the field of distributed point strain measurement. One of the most interesting properties of these sensors is the presumed linear relationship between the strain and the peak wavelength shift of the FBG reflected spectra. However, subjecting sensors to a non-uniform stress field will in general result in a strain estimation error when using this linear relationship.
In this paper we propose a new strain estimation algorithm that accurately estimates the mean strain value in the case of  smooth \color{black} non-uniform strain distributions. To do so, we first introduce an approximation of the classical transfer matrix model, which we will refer to as the approximated transfer matrix model (ATMM). This model facilitates the analysis of FBG reflected spectra under arbitrary strain distributions, particularly by providing a closed-form approximation of the side-lobes of the reflected spectra. \color{black} Based on this new formulation, we derive a maximum likelihood estimator of the mean strain value. The algorithm  is validated using both computer simulations and experimental FBG measurements. Compared to state-of-the-art methods, which typically introduce errors of tens of microstrains, the proposed method is able to compensate for this error. In the typical examples that were analysed in this study, mean strain errors of around $60 \mu\varepsilon$ were compensated.

 \blfootnote{This paper was submitted for review on January 31, 2018. \\ \indent
Aydin Rajabzadeh is with the Circuits and Systems Group of the Electrical Engineering Faculty and also with the Structural Integrity and Composites Group of the Aerospace Faculty of Delft University of Technology, Delft, 2628 CD The Netherlands (e-mail: a.rajabzadehdizaji@tudelft.nl).\\ \indent
Richard Heusdens and Richard C. Hendriks are with the Circuits and Systems Group of the Electrical Engineering Faculty of Delft University of Technology, Delft, 2628 CD The Netherlands (e-mail: r.heusdens@tudelft.nl and r.c.hendriks@tudelft.nl).\\ \indent
Roger M. Groves is with the Structural Integrity and Composites Group of the Aerospace Faculty of Delft University of Technology, Delft, 2629 HS The Netherlands (e-mail: r.m.groves@tudelft.nl).
}
\end{abstract}



\IEEEpeerreviewmaketitle

\section{Introduction}

\IEEEPARstart{F}{ibre} Bragg grating (FBG) sensors have shown to be one of the most robust and versatile sensors in strain sensing applications. With the possibility of multiplexing several sensors on a single optical fibre and their immunity to electromagnetic interferences, FBG sensors are an ideal candidate for applications in adverse environments such as the oil and gas industries or aviation~\cite{childers2001use}. Because of their small diameter, they can be embedded inside composite structures to acquire internal strain measurements \cite{torres2011analysis,kuang2001embedded} and can also be used for damage detection purposes including delamination and matrix cracks~\cite{de2008health,okabe2000detection,takeda2007damage,takeda2005delamination}. In addition, strain (or temperature) changes along the length of the sensor are directly related to the peak wavelength shift of the reflected spectra of FBG sensors (the intensity of the reflected light), making strain measurements straightforward. This relationship will be described in detail in Section~\ref{Theory}. However, FBG sensors have several complications which need to be addressed properly. The first problem is that the strain and temperature jointly contribute to the peak wavelength shift. In order to alleviate this problem, there have already been several extensive studies, either using multiple sensors with compensation for local temperature variations \cite{kersey1997fiber,montanini2007simultaneous}, or taking advantage of different grating structures such as tilted FBG sensors \cite{caucheteur2005simultaneous,alberto2010three}, multiple chirped FBG sensors \cite{gan2011high}, and birefringent FBG sensors with multiple independent peaks \cite{abe2004superimposed}. 

The second problem with FBG strain measurements is the sensor response under non-uniform strain distributions. Under uniform strain distributions, the relationship between the strain and the peak wavelength shift of the FBG reflection spectrum is linear. However, this does not necessarily hold under non-uniform strain distributions. Such non-uniformities could be the result of embedding the sensor in or surface mounting on a structure, when a non-uniform stress is applied. For instance, an embedded FBG sensor in a composite structure in the vicinity of a crack or delamination defect experiences different non-uniform strain distributions and, therefore, responds with rather distinguishable reflected spectra \cite{rajabzadeh2017classification,okabe2000detection,mizutani2011multi,takeda2007damage,takeda2005delamination}. The resulting non-uniform strain distributions under these circumstances could most possibly complicate strain measurements using FBG sensors.

Despite the fact that in many practical situations the strain distribution is non-uniform, existing algorithms for estimating the mean strain value are still based on the 
shift of the peak wavelength.
When the sensor is subject to non-uniform strain fields, however, each segment of the sensor will experience a different strain, resulting in different peak wavelength shifts
along the length of the sensor. As a 
consequence, looking to the (global) shift of the peak wavelength will in general lead to an estimation error in the mean strain value. 
In this paper we will propose a new algorithm that accurately estimates the mean strain value when the strain distribution is non-uniform. We will show that  
the mean strain value is related to the average shift of the peak wavelength along the length of the sensor and that this information can be found in the side lobes of
the reflected spectra. In order to analyse the reflected spectra we will present a model that describes the interaction of the forward and backward electric wave propagation 
between consecutive segments, which is an approximation of the widely used transfer matrix model (TMM). This approximated transfer matrix model (ATMM)  enables us to accurately 
find the average wavelength shift.  The codes for the approximated transfer matrix model and the mean strain estimation algorithm can be accessed online\footnote{ http://cas.tudelft.nl/Repository/}.

\color{black}

\section{Background}\label{Theory}
FBG sensors are spatially modulated patterns of refractive index changes in the core of optical fibres that act as a mirror for certain wavelengths. The linear relationship between the peak wavelength of the reflected spectra, usually referred to as the Bragg wavelength and denoted by $\lambda_B$, and the grating period $\Lambda$ of the FBG is given by \cite{hill1997fiber,melle1992passive}
\[
\lambda_B=2n_{\rm eff}\Lambda \,,
\]
where $n_{\rm{eff}}$ is the effective refractive index of the core. Under uniform strain (or temperature) fields, say $s$, the change in the grating period is constant along the length of the sensor, resulting in a shift $\Delta \lambda_{B}$ of the Bragg wavelength. That is,
\begin{equation}
\Delta \lambda_{B} = k_s  s,
\label{eq:ks}
\end{equation}
where the constant $k_s$ depends on the physical properties of the sensor\footnote{ For the sensors used in this study, based on the datasheet, $k_s=1.209\times 10^{-3} {\rm nm}/ \mu\varepsilon$.}.
In theory, this is without any change in the morphology of the reflected spectra, and the peaks will either be shifted towards shorter wavelengths (under compression) or longer wavelengths (under tension). Therefore, under a uniform strain field, (\ref{eq:ks}) results in an accurate measurement of the mean strain value over the length of the sensor. However, perfect uniform strain distributions are unlikely in practice. The more common strain fields are non-uniform and result in nontrivial overall reflected spectra, possibly asymmetric and having multiple peaks, leading to an error in the mean strain estimation. In order to compensate for this error, the full spectrum of the signal should be analysed, for which the transfer matrix method has been shown to be a proper tool. 

\subsection{Transfer matrix model} \label{TMM}
\label{sec:tmm}

Consider the case of a non-uniform strain distribution over the length of the FBG sensor.
We will divide the length of the sensor into a series of small segments of length $\Delta z$, where 
$\Delta z$ is taken sufficiently small such that each of these segments has approximately a uniform strain distribution. 
As a consequence, the strain distribution will affect each segment's grating period differently. To be more precise, let $s_i$ denote the strain field of segment $i$.
The Bragg wavelength shift $\Delta \lambda_{B_i}$ of segment $i$ is then given by
\[
\Delta \lambda_{B_i} = \lambda_{B_i} - \lambda_B = k_s s_i,
\]
where $\lambda_{B_i}$ is the Bragg wavelength of segment $i$. As a consequence, the mean strain distribution, denoted by $\bar{s}$,
satisfies
\begin{equation}
 k_s \bar{s} = \bar{\lambda}_B - \lambda_B,
\label{eq:smean}
\end{equation}
where $\bar{\lambda}_B = \frac{1}{M} \sum_i\lambda_{B_i}$ is the mean Bragg wavelength of the sensor. Hence, when the sensor is subject to non-uniform strain fields, each segment of the model experiences a different Bragg wavelength shift and 
$\bar{\lambda}_B$ does not necessarily correspond to the peak wavelength anymore.

In order to analyse the FBG reflected spectra under an arbitrary grating distribution, we will make use of the transfer matrix model (TMM) \cite{yamada1987analysis}. The TMM models the interaction of the forward and backward electric wave propagation between consecutive segments, where it is assumed that the length of the individual segments $\Delta z$ satisfies $\Delta z \gg \Lambda$. Let $A_i$ and $B_i$ denote the forward and backward propagating waves in segment $i$, respectively (see Fig.~\ref{fig:wave}). In this model, it is assumed that at the end of the final segment, there will be a full transmission of the incident wave ($A_0=1$) and no reflection from further along the optical fibre ($B_0=0$). 
\begin{figure}[t]
\center
\includegraphics[width=.45\textwidth]{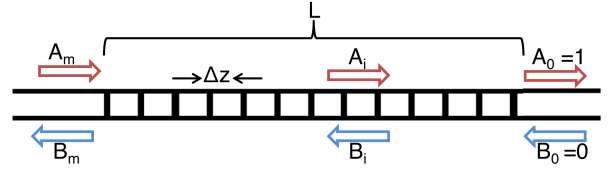}
\caption{Schematic view of the FBG structure.}\label{fig:wave}
\end{figure}

In each segment, the relation between the forward and the backward propagating waves can be described by the following relation \cite{yamada1987analysis}
\[
\begin{pmatrix}
  A_{i} \\
  B_{i} \\ 
 \end{pmatrix}
 =F_i 
 \begin{pmatrix}
  A_{i-1} \\
  B_{i-1} \\ 
 \end{pmatrix},
\]
where 
{\small
\begin{equation}
F_i=
 \begin{pmatrix}
  \cosh(\gamma_i \Delta z)-j\frac{\Delta \beta_i}{\gamma_i}\sinh(\gamma_i \Delta z) \,\, &\hspace{-0em}  -j \frac{\kappa_i}{\gamma_i}\sinh(\gamma_i\Delta z) \\
  \hspace{-6em} j \frac{\kappa_i}{\gamma_i}\sinh(\gamma_i\Delta z) \,\, & \hspace{-6em} \cosh(\gamma_i \Delta z)+j\frac{\Delta \beta_i}{\gamma_i}\sinh(\gamma_i \Delta z) \\
 \end{pmatrix}, \label{eq:Fi}
\end{equation}}

\vspace{-\baselineskip}
\noindent
is the tranfer matrix of segment $i$ having elements $F_{i_{11}}$, $F_{i_{12}}$, $F_{i_{21}}$ and $F_{i_{22}}$.
Note that  $F_{i_{11}} = F^*_{i_{22}}$ and  $F_{i_{12}} = F^*_{i_{21}}$,  where the superscript $\mbox{}^*$ denotes complex conjugation.
In (\ref{eq:Fi}), $\kappa_i$ is the coupling coefficient between forward and backward waves, $\Delta \beta_i = 2 \pi n\textsubscript {eff} (\frac{1}{\lambda}-\frac{1}{\lambda \textsubscript B_i})$ is the difference between the propagation constants in the longitudinal (or $z$) direction, $\gamma_i = \sqrt {\kappa_i ^ 2-\Delta \beta_i ^ 2} $, and $\lambda$ is a given wavelength under investigation.  
Assume that the total number of segments considered in the model is $M$. After calculating the matrix $F_i$ for each segment (for a given wavelength), the relationship between the backward and forward electric wave amplitudes between the $0^{th}$ and the $M^{th}$ segment is given by

\[
\begin{pmatrix}
  A_M \\
  B_M \\
 \end{pmatrix}
 =
 F
 \begin{pmatrix}
  A_{0} \\
  B_{0} \\
 \end{pmatrix},
\]
in which
\[
F= \prod_{i=1}^M F_i.
\]
The reflected spectrum, say $R(\lambda)$, is then determined as
\[
R(\lambda)={\left|\frac{B_M}{A_M}\right|}^2 = {\left|\frac{F_{21}}{F_{11}}\right|}^2,
\]
where $F_{21}$ and $F_{11}$ are  entries of the composite matrix $F$. In the next section, an approximation of the transfer matrix model will be 
introduced  which significantly simplifies the analysis of FBG reflected spectra. 

\section{Approximated transfer matrix model} 
\label{ATMM}
\label{sec:ATMM}

In this section we will show that, for sufficiently small $\Delta z$, the TMM can be accurately approximated, resulting in a model that facilitates the analysis of non-uniform strain fields. We will refer to this approximated model as the approximated transfer matrix model (ATMM). Note that the length $\Delta z$ can be chosen to be arbitrarily small (as long as $\Delta z \gg \Lambda$) and that this choice has no effect on the physical properties of the sensor. \\
\indent
Suppose that $\Delta z$ is sufficiently small. We will show that in this case the matrix $F_i$ can be approximated as
\begin{equation}
F_i \approx
 \begin{pmatrix}
  e^{-j \Delta \beta_i \Delta z} &
  -j\kappa_i \Delta z\sinc(\Delta \beta_i \Delta z)\\
  j\kappa_i \Delta z\sinc(\Delta \beta_i \Delta z) &  e^{j \Delta \beta_i \Delta z} \\
 \end{pmatrix}, \label{eq:Fi_hat}
\end{equation}
where $\sinc(\cdot)=\frac{\sin(\cdot)}{(\cdot)}$.\\
\indent
To show under what conditions the approximations hold, an element-wise comparison will be made between the formulations in (\ref{eq:Fi}) and  (\ref{eq:Fi_hat}). Recall that $\Delta \beta_i = 2 \pi n\textsubscript {eff} (\frac{1}{\lambda}-\frac{1}{\lambda_{B_i}})$ and $\gamma_i = \sqrt {\kappa_i ^ 2-\Delta \beta_i ^ 2} $, where $\lambda_{B_i}$ is the Bragg wavelength of segment $i$. In the near infrared wavelength range, which is the region of interest for almost all FBG sensors, we have ${\Delta \beta_i}^2 \gg {\kappa_i}^2$, and thus $\gamma_i \approx j|\Delta  \beta_i|$, except for a small wavelength range centred around $\lambda_{B_i}$ (a few tens of picometres which depends on the production and reflectivity levels of the FBG sensor) where the values of $\kappa_i$ and $\Delta\beta_i$ are of the same order of magnitude. In order to make the analysis of the formulation easier, we divide the wavelength range into the two above mentioned regions, as different strategies are needed to verify the correctness of the approximation for these regions.
\vspace{.5\baselineskip}\\
\noindent
\textbf{Region 1}: In this region $|\Delta \beta_i|\gg \kappa_i$. As already mentioned, for this region $\gamma_i \approx j|\Delta \beta_i|$, and therefore, the first term of $F_{i_{11}}$ in (\ref{eq:Fi}), i.e., $\cosh(\gamma_i\Delta z)$, can be approximated as\\
\[
\cosh(\gamma_i \Delta z)\approx \cosh(j|\Delta \beta_i| \Delta z)=\cos(\Delta \beta_i \Delta z),
\]
where we omit the absolute value since cos is an even-symmetric function. The second term of $F_{i_{11}}$ in (\ref{eq:Fi}) can in this wavelength region be approximated as 
\begin{align*}
-j \frac{\Delta \beta_i}{\gamma_i}\sinh(\gamma_i \Delta z) &\approx - {\rm sign}(\Delta \beta_i) \sinh(j|\Delta \beta_i| \Delta z) \\
&=-j\sin(\Delta \beta_i \Delta z),
\end{align*}
where we omit the absolute value since sin is an odd-symmetric function. Combining
these relations gives the required result for $F_{i_{11}}$ in (\ref{eq:Fi_hat}). The remaining term ($F_{i_{21}}$) can in this wavelength region be approximated as
\begin{align*}
 j\frac{\kappa_i}{\gamma_i}\sinh(\gamma_i\Delta z)\! &\approx \! j\frac{\kappa_i}{|\Delta \beta_i|}\sin(|\Delta \beta_i| \Delta z)\!  \\
 &= \!j\kappa_i \Delta z \, \sinc(\Delta \beta_i \Delta z).
\end{align*}
\vspace{-.5\baselineskip}\\
\noindent
\textbf{Region 2}: In this wavelength range, the values of $\kappa_i$ and $\Delta \beta_i$ are of the same order of magnitude so that $|\gamma_i\Delta z| \ll 1$. Hence, the Taylor series expansion for the first term of $F_{i_{11}}$ in (\ref{eq:Fi}) will be
\[
\cosh(\gamma_i \Delta z)= 1+\frac{1}{2}|\gamma_i \Delta z|^2+\frac{1}{4!}|\gamma_i \Delta z|^4+ \cdots . 
 \label{eq:approx21}
\]
In order to keep the representation of the ${F}_i$ matrices consistent in both regions, we replace $\cosh(\gamma_i \Delta z)$ by $\cos(\Delta \beta_i \Delta z)$, which leads in Region 2 to an absolute error of 
\begin{align}
&|\cosh(\gamma_i \Delta z) - \cos(\Delta \beta_i \Delta z)| \\  \nonumber
&=\frac{1}{2}\Delta z^2\left(|\Delta \beta_i|^2+ |\gamma_i|^2\right)+\frac{1}{6!}\Delta z^6\left(|\Delta \beta_i|^6+ |\gamma_i|^6\right) + \cdots , \nonumber
\end{align}
\noindent
which is negligible when $\Delta z$  is sufficiently small\footnote{ In this study, using computer simulations, it was seen that for $\Delta z \leq 0.001 \, \textrm{m}$ (corresponding to $M\geq 10$ for a sensor of length $1\, \rm{cm}$), the relative error of the amplitude of the reflected spectra was less than $0.05\%$.}. Along the same lines and using the Taylor series expansion of sin and sinh, the second term of $F_{i_{11}}$ is approximated as
\[
j\frac{\Delta \beta_i}{\gamma_i}\sinh(\gamma_i \Delta z)\approx j\frac{\Delta \beta_i}{\gamma_i}(\gamma_i  \Delta z) \approx j\sin(\Delta \beta_i \Delta z).
\]
\noindent
From this we see that the approximation of $F_{i_{11}}$ from (\ref{eq:Fi}) to (\ref{eq:Fi_hat}) can be argued to be also accurate in region 2. With respect to $F_{i_{21}}$ in (\ref{eq:Fi}), this can in a similar way be approximated by
\[
j\frac{\kappa_i}{\gamma_i}\sinh(\gamma_i\Delta z)\approx  j\kappa_i \Delta z
\approx j\kappa_i \Delta z \sinc(\Delta \beta_i \Delta z),
\]
\indent
which shows that the approximation of (\ref{eq:Fi}) by (\ref{eq:Fi_hat}) can be argued to be accurate also in region 2. Finally, we will rewrite (\ref{eq:Fi_hat}) in a slightly more convenient form by a variable substitution. As stated in Section \ref{TMM}, the difference between the propagation constants is $\Delta \beta_{i}=2\pi n_{\rm {eff}}(\frac{1}{\lambda}-\frac{1}{\lambda_{B_i}})$. Therefore, the argument $\Delta \beta_i \Delta z$ in (\ref{eq:Fi_hat}) can be expressed as $\Delta \beta_i \Delta z=\alpha-\alpha_i$, where 
\begin{equation}
\alpha=\frac{2\pi n_{\rm{eff}}\Delta z}{\lambda} \textrm{ and } \alpha_i=\frac{2\pi n_{\rm{eff}}\Delta z}{\lambda_{B_i}}.
\label{eq:alphas}
\end{equation} 
With this, (\ref{eq:Fi_hat}) can be rewritten as
\begin{equation}
{F_i}=
 \begin{pmatrix}
  e^{-j (\alpha-\alpha_i)} &
  -j\kappa_i \Delta z\sinc(\alpha-\alpha_i)\\
  j\kappa_i \Delta z\sinc(\alpha-\alpha_i) &  e^{j (\alpha-\alpha_i)} \\
 \end{pmatrix},
 \label{eq:Fi_hatf}
\end{equation}
which will be the basis of the analyses in this paper.

\section{Mean strain estimation} \label{meanerrorcomp}

To get a better understanding of the problem with the inaccuracy of  (\ref{eq:ks}) under non-uniform strain distributions, this section will start with an example in which the simulation results of an FBG sensor under non-uniform strain distributions are presented. Suppose the strain distribution of Fig.~\ref{fig:dists1}, with an average strain of $253.5 \mu \varepsilon$, is applied over the length of the FBG sensor. The resulting FBG reflected spectrum will have a non-symmetrical shape and is depicted in Fig.~\ref{fig:dists2}.
\begin{figure*}[h]
\begin{subfigure}{.5\textwidth}
  \centering
  \includegraphics[width=1\textwidth]{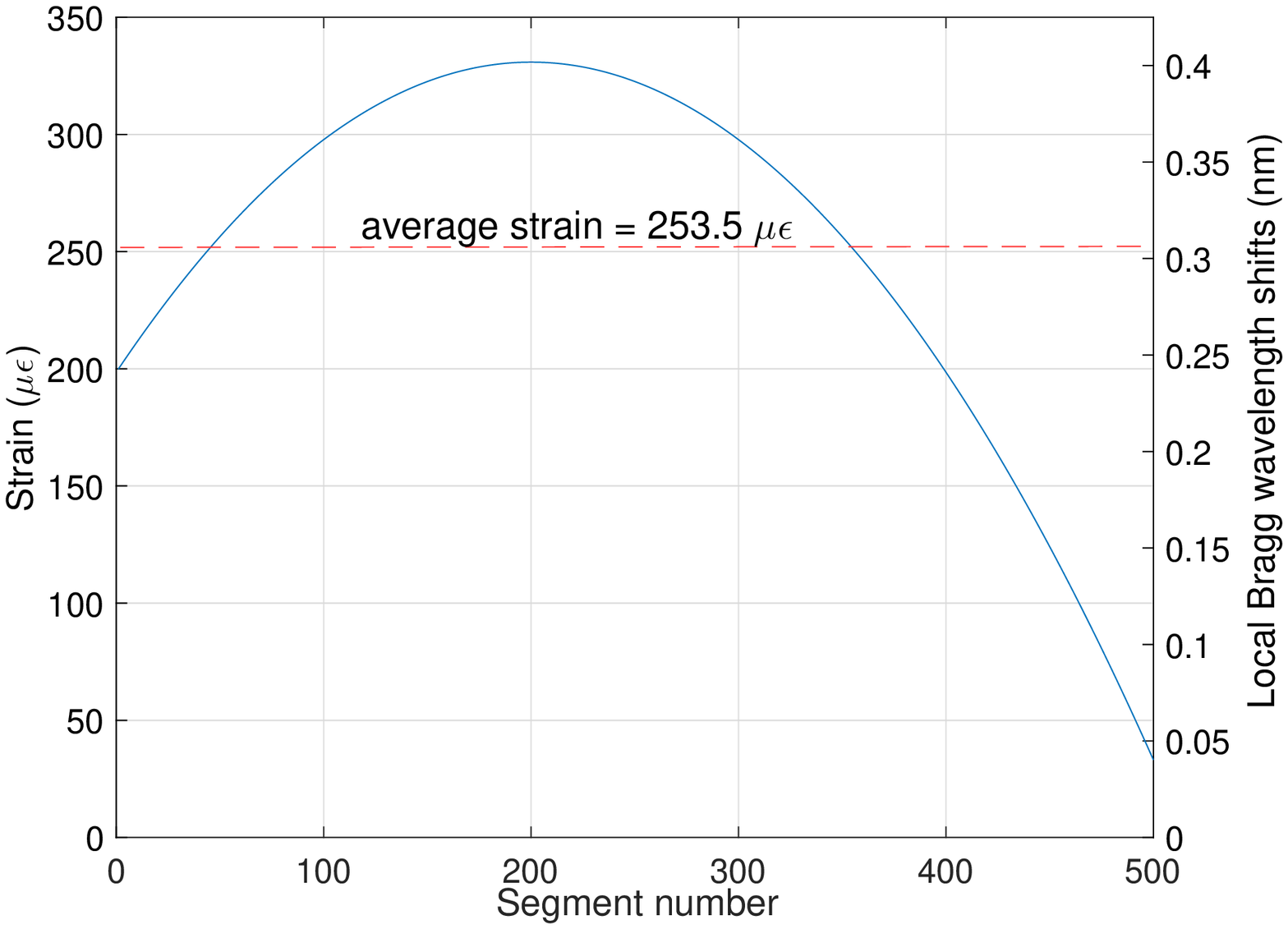}
  \caption{}
  \label{fig:dists1}
\end{subfigure}%
\begin{subfigure}{.5\textwidth}
  \centering
  \includegraphics[width=1\textwidth]{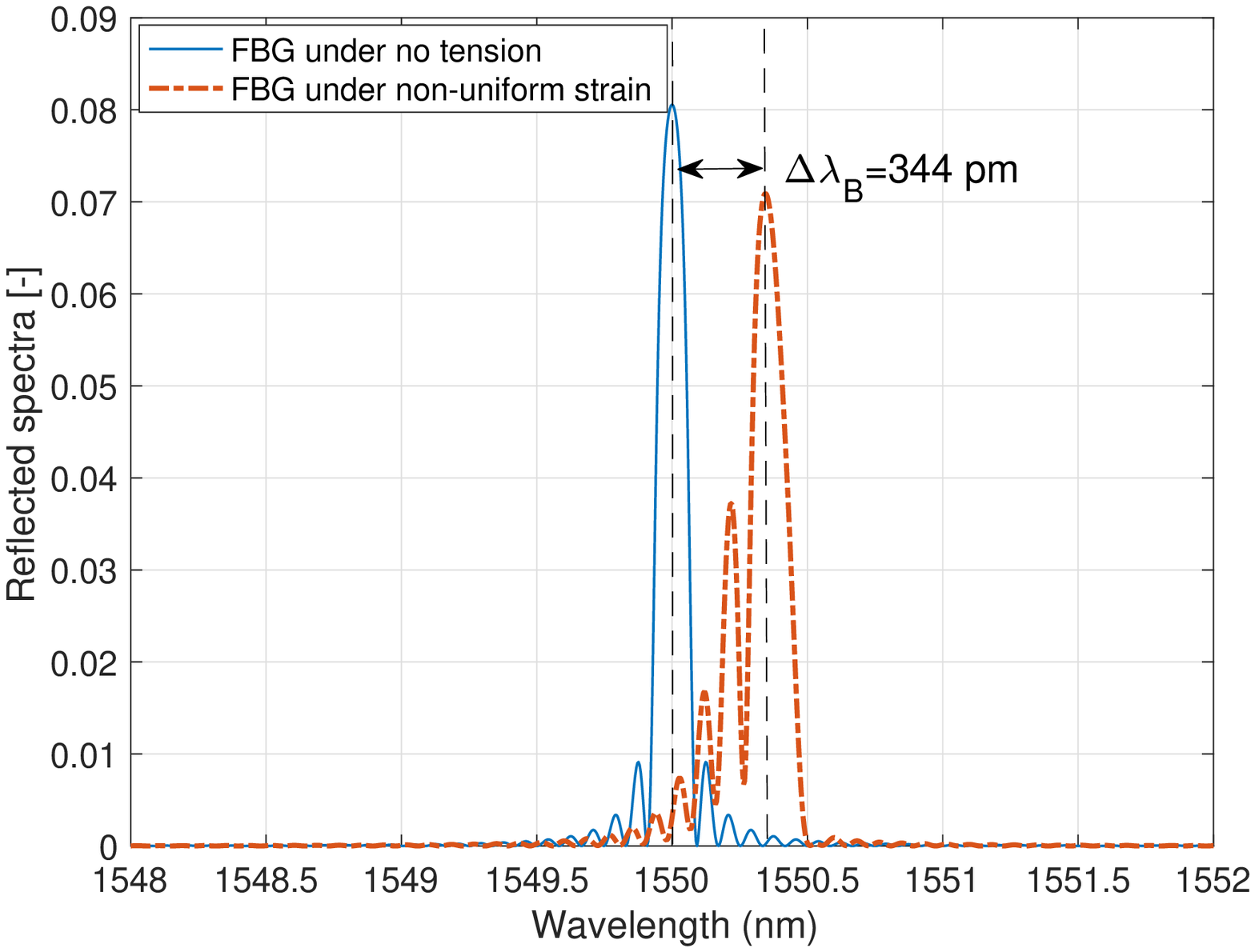}
  \caption{}
  \label{fig:dists2}
\end{subfigure}
\caption{(a): A non-uniform strain distribution (b): The resulted FBG reflected spectra.}
\end{figure*}
According to Fig.~\ref{fig:dists2}, the distance between the peak wavelength of the reflected spectra of the undisturbed and strained FBG sensor is $344 \, \rm{pm}$, which corresponds to a (mean) strain of $344/k_s=284.5 \, \mu \varepsilon$. Hence, in this example, an error of approximately $31\, \mu \varepsilon$ is introduced which needs to be compensated.

By inspection of \eqref{eq:smean} we conclude that when subject to a non-uniform strain distribution, instead of finding the peak wavelength of the reflected spectrum of the sensor, we need to find the average  Bragg wavelength $\bar{\lambda}_B$ of the sensor. 
In the next subsection, we will see that $\bar{\lambda}_B$ can be found by inspection of the  side lobes of the FBG reflected spectra and does not necessarily correspond to the peak wavelength. To do so,  we will use the ATMM to derive a closed-form approximation for the side lobes.
\color{black}
\subsection{Closed-form approximation of the side lobes} \label{sidepeaks}

In what follows, we assume the coupling coefficient ($\kappa_i$) to be constant and equal to $\kappa$ throughout the length of the sensor. The reason for this assumption is that first, the coupling coefficient does not affect the oscillation frequency of the FBG reflected spectra in neither main nor side peaks, and second, due to the fact that it often has only a small variation along the length of the sensor, its effect on the amplitude is negligible. 

Suppose there is an arbitrary non-uniform strain distribution over the length of the FBG sensor. Each segment $i$ of the FBG sensor undergoes a local strain $s_i$, resulting in an asymmetric overall reflected spectrum. Multiplying all $M$ approximated transfer matrices ${F}_i$ defined in (\ref{eq:Fi_hatf}), resulting in the composite matrix ${F}$, entry ${F}_{11}$ will have the form 
\begin{align}
& \!\!\!{F}_{11}= e^{-j(M\alpha-\sum_{i=1}^{M} \alpha_i)} \nonumber \\
& +{\sum_{n=1}^{\floor*{\frac{M}{2}}}\sum_{l=1}^{\binom {M} {2n}} \{ \prod_{i \in x_{l}}{} \!\!\! (-1)^n \kappa \Delta z \sinc(\alpha-\alpha_{i}) \prod_{i \in x_{l}^C}{}  e^{(-1)^{v}j(\alpha-\alpha_{i})}\}},
\label{eq:Ftot11}
\end{align}
where $x_{l} \in X_{2n}$ with $X_{2n}$ being the set of all possible combinations of $2n$ numbers taken from the set $\Omega = \{1,2,...,M\}$. As an example, assuming $n=1$, we will have $X_{2}=\{\{1,2\} , \{1,3\}, ... \{M-1,M\}\}$. Also, $x_{l}^C=\Omega \setminus x_{l}$ is the complement of the set $x_{l}$ in $\Omega$, and $v \in {\{0,1\}}$ which depends on the set $x_{l}$. Similarly, we find that

\begin{align}
 {F}_{21}\! &= \sum_{i=1}^{M} \, \kappa \Delta z \sinc(\alpha-\alpha_i)e^{\!\!-j\left(\rule[3mm]{0mm}{0mm} \right. \!\! (M-2i+1)\alpha+\sum \limits_{k<i}{\alpha_k}-\sum \limits_{k>i}{\alpha_k} \!\! \left. \rule[3mm]{0mm}{0mm} \right)} \nonumber \\
&\hspace{-6mm}+{\sum_{n=1}^{\floor*{\frac{M}{2}}}\!\!\sum_{l=1}^{\binom {M} {2n+1}} \!\! \{ \prod_{i \in y_{l}}{} \!\! (-1)^n \kappa \Delta z \sinc(\alpha-\alpha_{i})\!\!\prod_{i \in y_{l}^C}{} \!\!e^{(-1)^{v}j(\alpha-\alpha_{i})}\}},
\label{eq:Ftot21}
\end{align}
where $y_{l} \in Y_{2n+1}$ with $Y_{2n+1}$ being the set of all possible combinations of $2n+1$ numbers taken from $\Omega$.
By inspection of (\ref{eq:Ftot11}), it can be seen that when the sinc terms are sufficiently damped, the dominant term will be the first exponential, which has magnitude $1$. As a consequence
\begin{equation}
{R}(\lambda)\approx \left|{F}_{21}\right|^2 ,\,\,\,\,\,\,\textrm{for} \,\, \forall i \in \Omega, \,\, |\lambda-\lambda_{B_i}| > \lambda_{\rm th},\\
\label{eq:Rsimp}
\end{equation}
where $\lambda_{\rm th}>0$ is a threshold wavelength for which \eqref{eq:Rsimp} holds. Note that the condition $|\lambda-\lambda_{B_i}| > \lambda_{\rm th}$ is approximately the same as $|\alpha-\alpha_i| > \alpha_{\rm th}$. 
Using similar arguments, it can be shown that the dominant terms in (\ref{eq:Ftot21}) are given by those in the first summation, as the second summation contains products of sinc functions whose amplitudes are small when $|\lambda-\lambda_B| > \lambda_{\rm th}$.
For notational convenience, let $\xi_i = \frac{\kappa_i \Delta z}{2j(\alpha-\alpha_i)}$. With this, we have
\begin{align}
F_{21} &\approx \sum_{i=1}^{M} \, \kappa \Delta z \sinc(\alpha-\alpha_i)e^{\!\!-j\left(\rule[3mm]{0mm}{0mm} \right. \!\! (M-2i+1)\alpha+\sum \limits_{k<i}{\alpha_k}-\sum \limits_{k>i}{\alpha_k} \!\! \left. \rule[3mm]{0mm}{0mm} \right)}\nonumber \\
&\stackrel{(a)}{=} \sum_{i=1}^{M}\xi_i \left( \rule[4.5mm]{0mm}{0mm} \right. e^{\!\!-j\left(\rule[3mm]{0mm}{0mm} \right. \!\! (M-2i)\alpha+\sum \limits_{k\leq i}{\alpha_k}-\sum \limits_{k>i}{\alpha_k} \!\! \left. \rule[3mm]{0mm}{0mm} \right)} \nonumber \\
&\hspace{18mm}-  e^{\!\!-j\left(\rule[3mm]{0mm}{0mm} \right. \!\! (M-2(i-1))\alpha+\sum \limits_{k<i}{\alpha_k}-\sum \limits_{k\geq i}{\alpha_k} \!\! \left. \rule[3mm]{0mm}{0mm} \right)} \left. \rule[4.5mm]{0mm}{0mm} \right) \nonumber \\
&\stackrel{(b)}{=} \sum_{i=1}^{M-1} \left(\xi_i - \xi_{i+1}\right) e^{\!\!-j\left(\rule[3mm]{0mm}{0mm} \right. \!\! (M-2i)\alpha+\sum \limits_{k\leq i}{\alpha_k}-\sum \limits_{k>i}{\alpha_k} \!\! \left. \rule[3mm]{0mm}{0mm} \right)} \nonumber \\
&\hspace{15mm} + \xi_M e^{jM(\alpha-\bar{\alpha})} - \xi_1 e^{-jM(\alpha-\bar{\alpha})},
\label{eq:Ftot21exp}
\end{align} 
where ($a$) follows by applying Euler's equation and ($b$) is obtained by re-arranging terms. 
We have the following results.
\begin{lemma}
Let $\lambda_{B_i} = \bar{\lambda}_B + \Delta_i$ and $\overline{\Delta^2} = \frac{1}{M}\sum_i\Delta_i^2$. If $|\Delta_i| \ll \bar{\lambda}_B$ for all $i$, then
\begin{equation}
\hspace{-32mm}
\text{\rm 1)}\quad\quad \bar{\alpha} = \frac{2\pi n_{\rm eff}\Delta z}{\bar{\lambda}_B} + {\cal O}\left(\frac{\overline{\Delta^2}}{\bar{\lambda}_B^3}\right).
\label{eq:lemma}
\end{equation}
If, in addition, $|\lambda-\bar\lambda_B|>\lambda_{\rm th}$ and $|\lambda_{B_i}-\lambda_{B_{i+1}}| \ll \lambda_{\rm th}$ for all $i$, then
\vspace{.3\baselineskip}
\[
\hspace{-38mm}
\text{\rm 2)}\quad\quad  \left|\xi_i - \xi_{i+1}\right| \ll |\xi_1|, |\xi_M|.
\]
\label{lem:approx}
\end{lemma}
\begin{proof}
See Appendix A.
\end{proof}

Assuming that the conditions of Lemma~\ref{lem:approx} are satisfied, we conclude that
\begin{equation}
F_{21} \approx \xi_M e^{jM(\alpha-\bar{\alpha})} - \xi_1 e^{-jM(\alpha-\bar{\alpha})},
\label{eq:F21}
\end{equation}
so that
\begin{align}
R(\lambda)  &\stackrel{(a)}{\approx} 4{\rm Re}(\xi_1\xi^*_M) \sin^2\left(M(\alpha-\bar{\alpha})\right) + (\xi_M - \xi_1)^2 \nonumber \\
&= \frac{ (\kappa \Delta z)^2}{(\alpha-\alpha_1)(\alpha-\alpha_M)} \sin^2\left(M(\alpha-\bar{\alpha})\right) + (\xi_M - \xi_1)^2  \rule[6mm]{0mm}{0mm}\nonumber \\
&\stackrel{(b)}{\approx} (\kappa L)^2 \sinc^2\left(M(\alpha-\bar{\alpha})\right) + (\xi_M - \xi_1)^2, \rule[6mm]{0mm}{0mm}
\label{eq:sincapprox}
\end{align}
where ($a$) follows from \eqref{eq:F21} using elementary trigonometric identities and ($b$) follows from the presumption that small variations in the amplitude of the $\sin$ function in (\ref{eq:sincapprox}), caused by replacing the $\alpha_1$ and $\alpha_M$ terms by $\bar{\alpha}$ are negligible. Also note that $M\Delta z = L$, the length of the sensor. 

Some remarks are in place here. The assumptions for which the results of Lemma~\ref{lem:approx} hold are met in most practical scenarios. Indeed, in practice the deviations from the mean Bragg wavelength is less than a few nanometers. As an example, the strain distribution as depicted in Fig.~\ref{fig:dists1} gives rise to a maximum deviation of about 250 pm, which is, compared to the average Bragg wavelength ($\bar{\lambda}_B \approx 1550$ nm), three to four orders of magnitude smaller. Larger deviations are unrealistic in the sense that too large $\Delta_i$ will result in breaking the FBG sensor or require unrealistically long sensors. 
In addition, if we choose $\Delta z$ sufficiently small and assume that the strain distribution cannot change arbitrarily fast  along the length of the sensor, we will have $|\lambda_{B_i}-\lambda_{B_{i+1}}| \ll \lambda_{\rm th}$, where $\lambda_{\rm th}$ is in the order of 1-2 nm. By inspection of Fig.~\ref{fig:dists1}, if we choose $M=100$, the maximum difference between successive Bragg wavelengths $\lambda_{B_i}$ is approximately 5 pm, which is three orders of magnitude smaller than $\lambda_{\rm th}$. 

Coming back to the approximation \eqref{eq:sincapprox}, we see that the reflected spectrum can be approximated by a scaled (squared) sinc function having a possible offset. This approximation only holds in the wavelength range $|\lambda-\lambda_B|>\lambda_{\rm th}$.
To illustrate this approximation, 
Fig.~\ref{fig:sidelobes} compares the reflected spectrum of a simulated FBG sensor under the non-uniform strain field depicted in 
Fig.~\ref{fig:dists1}, along with its  approximation given by \eqref{eq:sincapprox}. It can be seen that this approximation does not hold for the main lobe, but does hold for the wavelength region for which $|\lambda-\lambda_B| > \lambda_{\rm th}$, where $\lambda_{\rm th}$ is in the order of 1-2 nm  which can be identified by setting a threshold level on the amplitude of the reflected spectra\footnote{ Here the threshold level was set at 1 percent of the peak amplitude ($20 \,dB$ difference in the amplitude).}.

\color{black}
\begin{figure}[t]
\center
\includegraphics[width=0.52\textwidth]{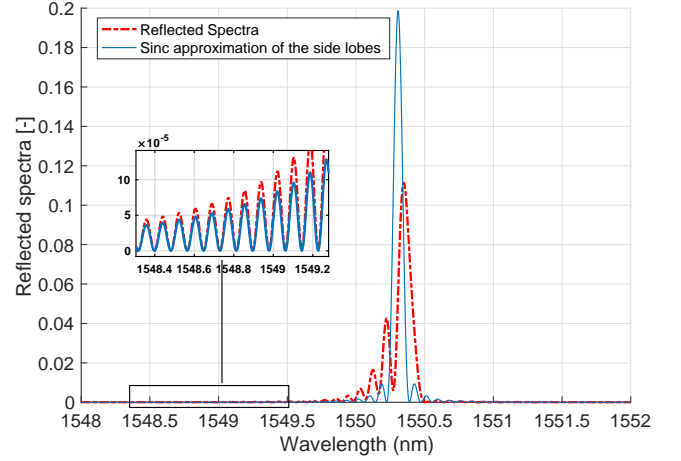}
\caption{FBG reflected spectra under an arbitrary non-uniform strain field (in blue) and the approximation for the side lobes in red.}\label{fig:sidelobes}
\end{figure}
%

As was shown in Section~\ref{sec:tmm}, it is $\bar{\lambda}_B$ that we need to estimate in order to compute the mean strain distribution $\bar{s}$
using  \eqref{eq:smean}. However, the result
of Lemma~\ref{lem:approx}, which gives the relation between  $\bar{\lambda}_B$ and $\bar{\alpha}$, shows that this can be accomplished by estimating $\bar{\alpha}$, which is the phase shift of the sinc approximation
of the reflected spectrum. 
A maximum likelihood estimator of 
$\bar{\alpha}$ is therefore given by the phase-shift value that maximises the cross-correlation between the observed reflected spectrum and  $\sinc^2\left(M(\alpha-\bar{\alpha})\right)$ \cite[p.\ 192]{kay:93}, where the correlation in this case is only taken over the range $|\lambda-\lambda_B| > \lambda_{\rm th}$. 
Note that since the cross-correlation is shift and scale invariant, neither any discrepancy between the magnitude of the reflected spectrum and its sinc approximation, as is present in the example shown in Fig.~\ref{fig:sidelobes}, nor a possible offset in the spectrum will affect the estimation of $\bar{\alpha}$. In addition, even when the first $M-1$ exponential terms in \eqref{eq:Ftot21exp} cannot be completely neglected, making the approximation \eqref{eq:F21} less accurate,  this will not have a significant impact on the estimation of  $\bar{\alpha}$ since the exponentials give low correlation with the sinc function (they have different oscillating frequencies), making the proposed method robust against inaccuracies in \eqref{eq:sincapprox}.  Also, this method is not subject to any additional spectral noise or errors compared with conventional FBG interrogation methods, and due to the filtering properties of the correlation function, it is robust against amplitude noise as well. It is worth mentioning that the proposed mean strain estimation method uses the information in the side-lobes of the reflected spectra, therefore, FBG sensors whose side-lobes are suppressed and have really small amplitudes (such as Gaussian or raised-cosine apodized FBG sensors) will not perform well with our methods.

\color{black}
\subsection{Practical considerations}
In practical scenarios the output of the FBG sensor is obtained  using an interrogator. As a consequence, the data available for processing are  samples of the reflected spectra, uniformly spaced in the $\lambda$-domain. Instead of performing the processing in the $\alpha$-domain, 
which would require a non-uniform re-sampling of the data, 
we could equally well perform the correlation in the wavelength domain directly. To see this, let 
$\lambda  = \bar\lambda_B + \Delta\lambda > \lambda_{\rm th}$, and assume that $|\Delta\lambda| \ll \bar\lambda_B$. This assumption is generally met in practice since $|\Delta\lambda|$ is in the order of a few nanometers, which is three orders of magnitude smaller than $\bar\lambda_B$. Moreover, let $\rho = 2\pi n_{\rm eff}\Delta z$. As a consequence, the $\alpha$ as defined in (\ref{eq:alphas}), can be rewritten as
\begin{align*}
\alpha &= \frac{\rho}{\bar\lambda_B + \Delta\lambda} \nonumber \\
&\stackrel{(a)}{\approx}  \frac{\rho}{\bar\lambda_B}\left(1 - \frac{\Delta\lambda}{\bar\lambda_B} \right) \nonumber \\
&\stackrel{(b)}{\approx} \bar\alpha -  \frac{\rho}{\bar\lambda_B^2} \Delta\lambda,
\end{align*}
where ($a$) is a first-order Taylor series approximation of $\alpha$ and ($b$) follows from \eqref{eq:lemma}. 
Hence, a linear change of the wavelength manifests itself as a linear change in $\alpha$, and as a consequence, uniform sampling of $\lambda$ will result in uniform sampling of $\alpha$ and vice versa, assuming $|\Delta\lambda| \ll \bar\lambda_B$. 

Although the above introduced method for retrieving $\bar\lambda$ will work for computer simulations, due to the presence of birefringence effects and other unwanted artefacts on real FBG measurements like non-longitudinal strains, the algorithm might lead to unwanted maxima in the cross-correlation function and therefore to incorrect phase retrieval. To overcome this problem in practical scenarios, we introduce a slight modification to the above proposed algorithm. Instead of computing the cross-correlation between the reflected spectrum and the sinc function directly,   we first align the two reflected spectra (recorded before and after applying the strain) based on their centre of mass~\cite{Askins1995}
\begin{equation}
\lambda_{B_c} = \frac{\int_{\mathbold{\lambda}} \lambda R(\lambda) d\lambda}{\int_{\mathbold{\lambda}}{R(\lambda)d\lambda}}
 \label{eq:centroid}
\end{equation}
where $\mathbold{\lambda}$ is the wavelength region that covers the reflection spectrum. This shift, $\Delta \lambda_{B_c}$, can be used to find a rough estimate of
$\bar{\lambda}$, say $\tilde{\lambda} = \lambda_B + \Delta\lambda_{B_c}$. 
\color{black}
After this, we calculate the final phase shift by maximising the cross-correlation of the side lobes of both observations over a small interval around 
$\tilde{\lambda}$, resulting in an additional phase shift
$\delta\lambda_B$. Experiments have shown that this modification results in more robust mean strain estimates and is illustrated in Fig.~\ref{fig:ps_out1} and \ref{fig:ps_out2}.
The modified algorithm is summarised in Algorithm \ref{alg:strain}. Note that, as mentioned in the introduction, existing algorithms for estimating the mean strain value are  based on the shift of the peak wavelength. That is, they estimate the mean strain value based on $\Delta\lambda_B$. In that sense, we could interpret $\delta\lambda_B$ as an error compensating term for methods based on peak wavelength alignment.

\begin{figure*}[h]
\begin{subfigure}{.5\textwidth}
  \centering
  \includegraphics[width=1.0\textwidth]{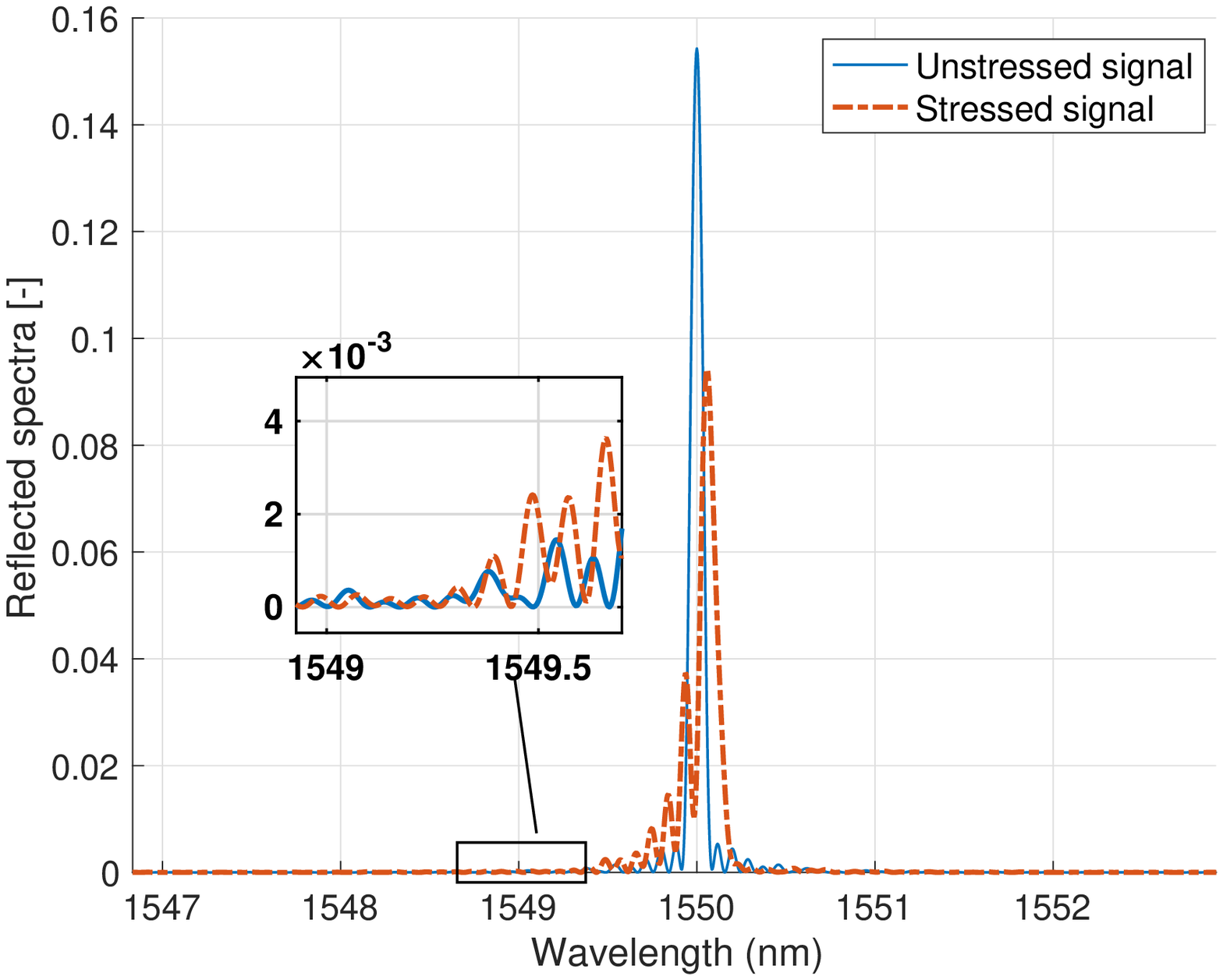}
  \caption{}
  \label{fig:ps_out1}
\end{subfigure}%
\begin{subfigure}{.5\textwidth}
  \centering
  \includegraphics[width=1.0\textwidth]{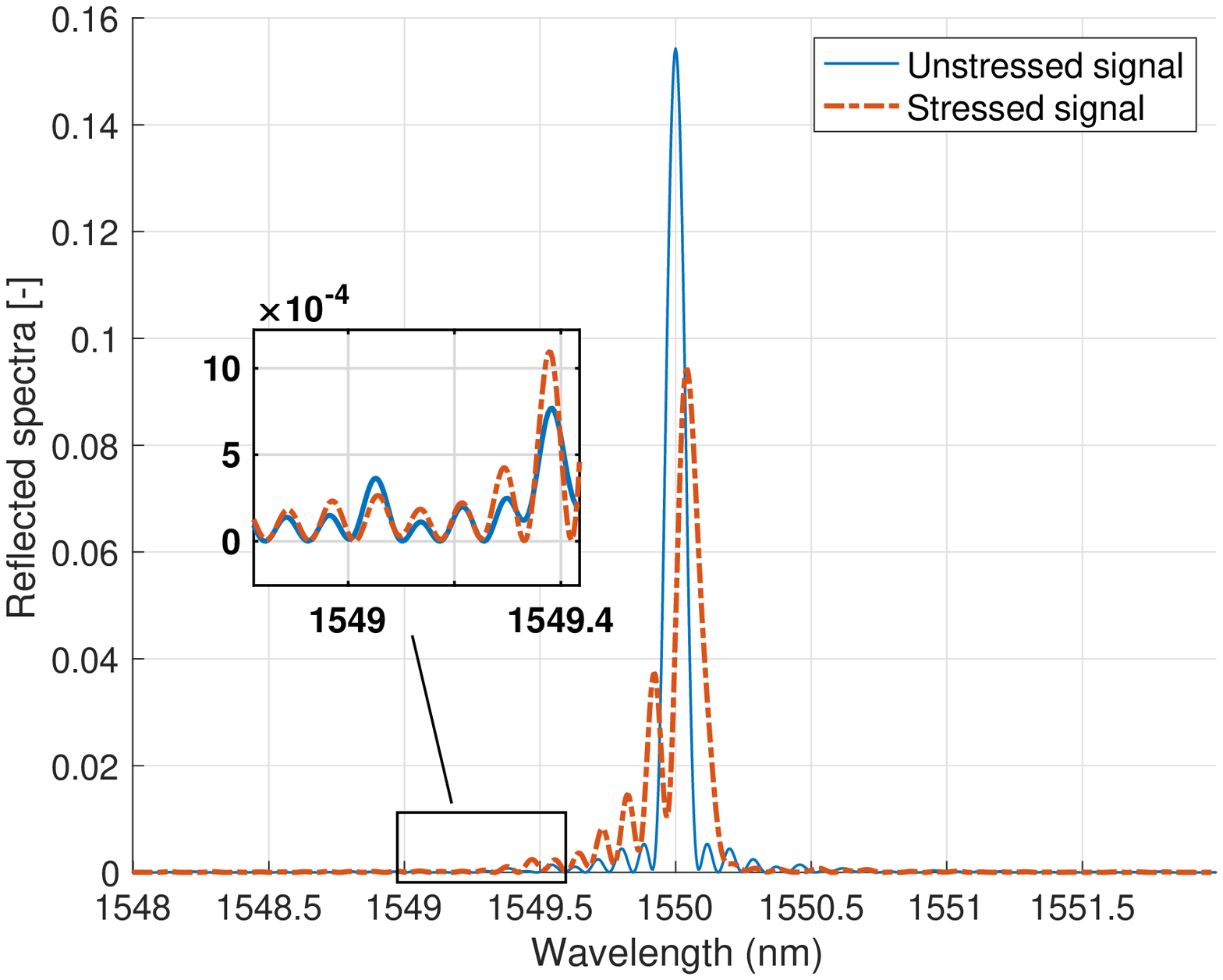}
  \caption{}
  \label{fig:ps_out2}
\end{subfigure}
\caption{The FBG reflected spectra of the unstrained sensor (blue) and the strained sensor (red) where (a) shows both spectra aligned with respect to their centre of mass and (b) with respect to maximising the cross-correlation of the side peaks.}
\label{fig:ps_out}
\end{figure*}
\begin{algorithm}[t]
\caption{Mean strain estimation \rule[4mm]{0mm}{0mm}\rule[-2mm]{0mm}{0mm}}\label{alg:strain} 
\begin{algorithmic}[1]
\State \rule[4mm]{0mm}{0mm}Align the  centre of mass of the reflected spectra of the FBG sensor measured with and without applying a (non-uniform) strain, thereby defining $\Delta\lambda_{B_c}$.  
\color{black}
\State \rule[5mm]{0mm}{0mm}Maximise the cross-correlation of the side lobes of both measurements over a small interval  with a typical value of $1 \, \textrm{nm}$ around \color{black} $\tilde{\lambda} = \lambda_B + \Delta\lambda_{B_c}$, resulting in an additional phase shift $\delta\lambda_{B}$.
\State \rule[5mm]{0mm}{0mm}Calulate the required phase shift $\bar{\lambda}_B -  \lambda_B = \Delta\lambda_{B_c} + \delta\lambda_B$.
\State \rule[5mm]{0mm}{0mm}Calculate the mean strain using (\ref{eq:smean}).\rule[-2mm]{0mm}{0mm}
\end{algorithmic}
\end{algorithm}

\section{Experimental Results}
In this section, we will present experimental results for non-uniform strain distributions on FBG sensors, obtained by both computer simulations and experimental FBG measurements.

In the first experiment, which is a computer simulation, the strain distribution of Fig.~\ref{fig:dists1} is applied over the length of the sensor, resulting in the FBG reflected spectrum of Fig. \ref{fig:dists2}, where we chose $M=500$.  In order to create a more realistic experiment, the sensor was assumed to have random fluctuations on the magnitude of the refractive index change, which results in random variations of the coupling coefficient along the length of the sensor, and eventually, in emergence of unwanted lower frequency harmonics on the side-lobes of the FBG reflection spectra.
Following the steps outlined in Algorithm~\ref{alg:strain}, we first align the reflected spectra with respect to their centre of mass, as shown in Fig.~\ref{fig:ps_out1}, and measure the shift $\Delta \lambda_{B_c}$. In the example at hand, this shift is given by $\Delta \lambda_{B_c}=296 \, \textrm{pm}$. The next step is then to shift the reflected spectrum such that the cross-correlation of the side lobes is maximised, as shown in Fig.~\ref{fig:ps_out2},
resulting in an additional phase shift of $\delta\lambda_B = +12 \, \rm{pm}$. After computing the final phase shift as $ \Delta\lambda_{B_c} + \delta\lambda_B = 308 \, \rm{pm}$ (step 3 of Algorithm~\ref{alg:strain}), the mean strain is computed using \eqref{eq:smean}, resulting in  $\bar{s} = 254.7 \,\mu \varepsilon$. It can be seen that the resulting estimation for the mean strain over the length of the sensor is quite close to the actual mean strain ($253.5 \,\mu \varepsilon$), where the proposed algorithm compensates for an error of around $30 \,\mu \varepsilon$ as compared to traditional strain measurement algorithms, see Section~\ref{meanerrorcomp}.

\color{black}
%
 
\noindent
 In the second experiment, which is another computer simulation, we investigate the performance of the algorithm when a non-uniform strain distribution is applied on a linearly chirped fibre Bragg grating (LCFBG) sensor. The designed sensor, again containing refractive index fluctuations, has a grating distribution with $\frac{d\lambda_B}{dz}=2.5 \, \textrm{nm/cm}$, and a length of $1 \,\textrm{cm}$. The strain distribution applied over the length of the sensor is shown in Fig.~\ref{fig:poly_strain}. The mean strain value in this experiment is $93 \, \mu\varepsilon$, but the traditional strain estimation algorithm in LCFBGs which is solely based on the shift of the centre of mass of the stressed and unstressed signals, show a wavelength shift of around $42.6 \, \textrm{pm}$ or a strain estimation of $35\, \mu\varepsilon$. Using our algorithm, the compensating shift that results from maximising the cross correlation of the side-lobes was $\delta \lambda_B= +72 \,\textrm {pm}$, which accounts for a mean strain error of $59.5 \, \mu\varepsilon$. The final mean strain estimation based on our algorithm is therefore, $94.5 \, \mu\varepsilon$. Fig.~\ref{fig:chirped_fbg} shows the reflected spectra of the stressed and unstressed LCFBG after the compensating shift. Note the synchronised side-lobes of the two reflected spectra.

\begin{figure*}[t]
\begin{subfigure}{.5\textwidth}
  \centering
  \includegraphics[width=0.97\textwidth]{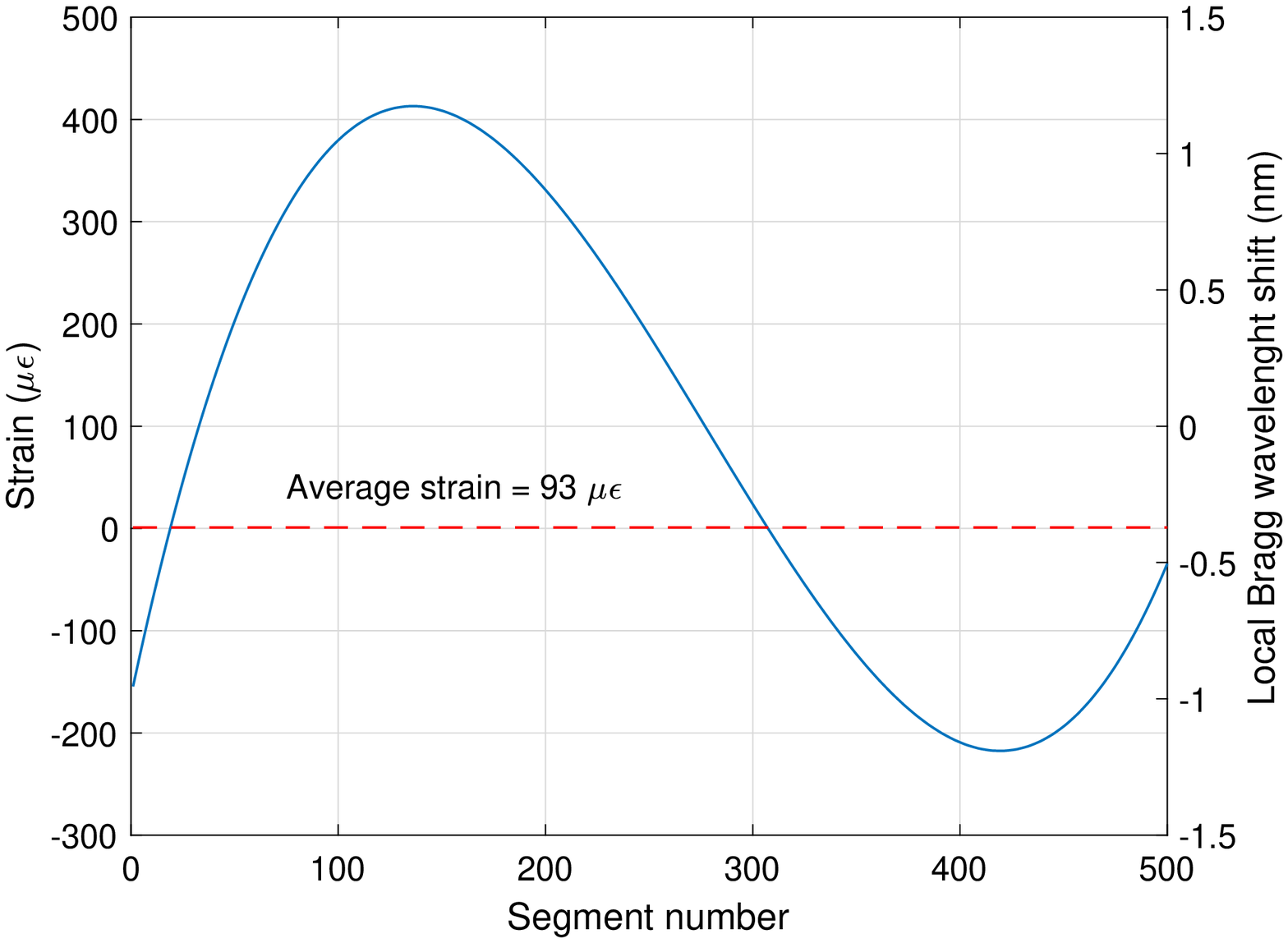}
  \caption{}
  \label{fig:poly_strain}
\end{subfigure}%
\begin{subfigure}{.5\textwidth}
  \centering
  \includegraphics[width=0.97\textwidth]{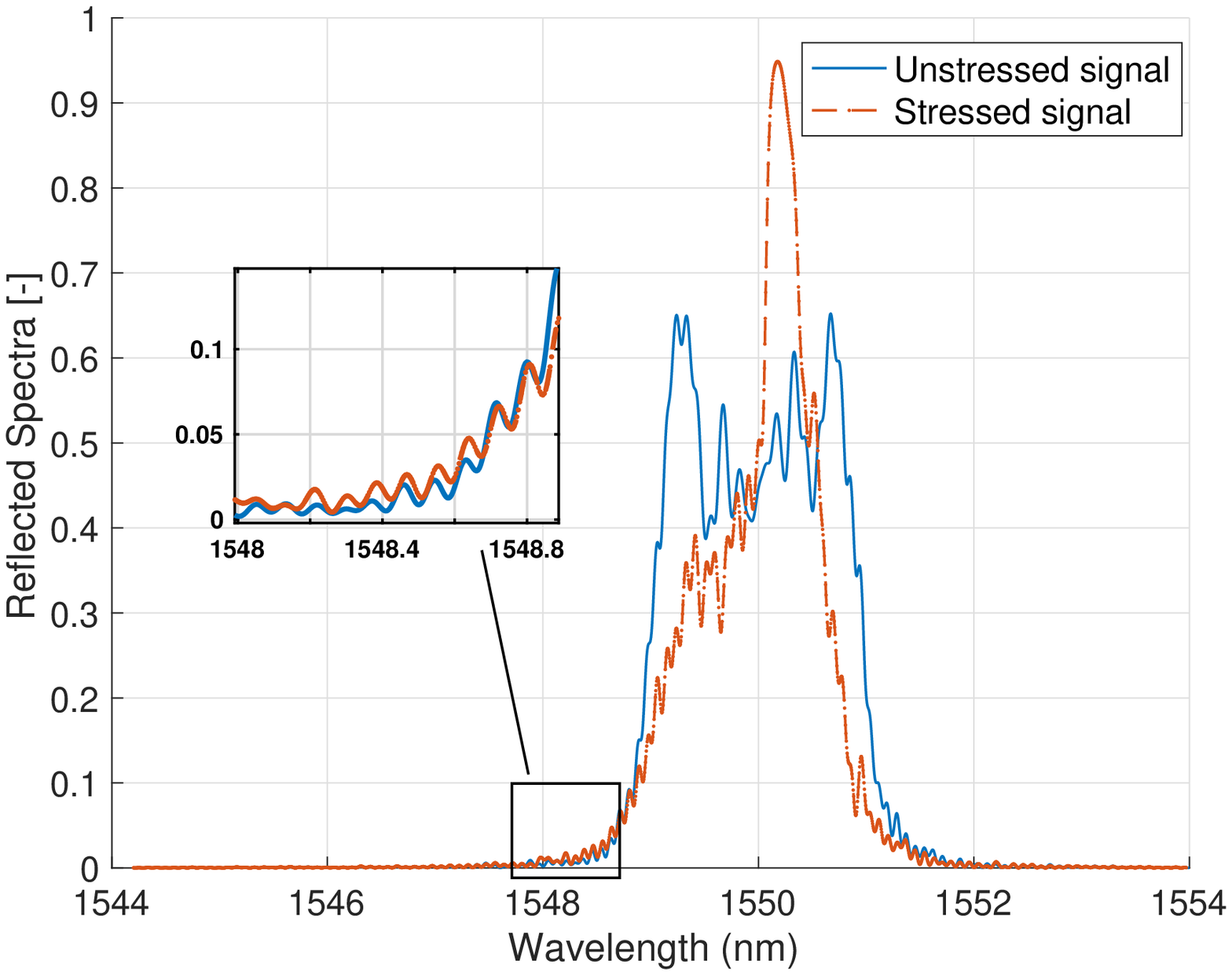}
  \caption{}
  \label{fig:chirped_fbg}
\end{subfigure}
\caption{(a): A polynomial type strain distribution, applied over the length of the LCFBG sensor. (b): The unstressed and stressed LCFBG reflected spectra are aligned with respect to maximising the correlation of their side peaks.}
\label{fig:chirping}
\end{figure*}

\color{black}
In the third experiment, the proposed method will be evaluated using experimental FBG measurements in a controlled laboratory environment. 
The FBG sensors used in this study are the LBL-1550-125 draw tower grating (DTG) type sensors (FBGS International NV). The length of the sensors are $10 \, \rm{mm}$, with a nominal Bragg wavelength of $1550.08 \, \rm{nm} $    and maximum reflectivity level of about $10 \%$. Although the algorithm would work at its best in sensors with higher reflectivity levels, we have two reasons for choosing DTG sensors with such low reflectivity levels in our experiments. The first reason is that DTG sensors have a much higher tensile strength compared to FBG sensors that are produced using traditional strip and recoating methods. They can easily endure static strains of more than $10000 \,\mu\varepsilon$ which can commonly occur in practical applications. The second reason is to show that our algorithm can equally well perform, even at such low reflectivity levels, provided that the sensor is interrogated with a high dynamic range interrogator.
\begin{figure*}[t]
\begin{subfigure}{.5\textwidth}
  \centering
  \vspace{8mm}
  \includegraphics[width=1\textwidth]{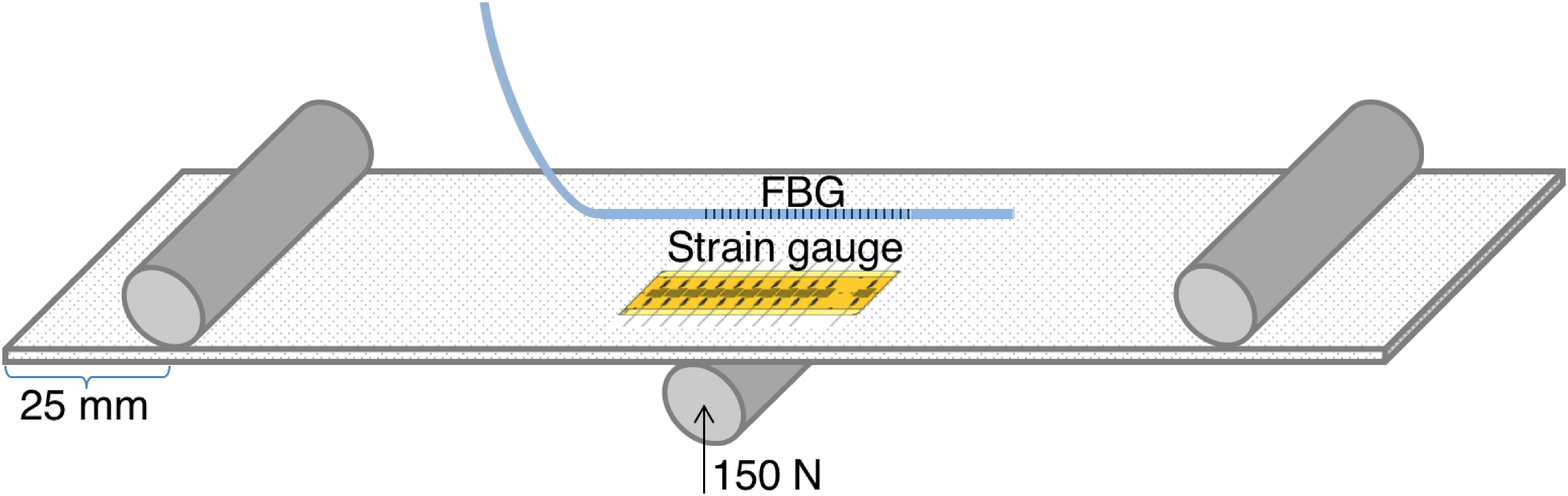}
  \caption{}
  \label{fig:sdist1}
\end{subfigure}%
\begin{subfigure}{.5\textwidth}
  \centering
  \hspace{1mm}
  \includegraphics[width=.8\textwidth]{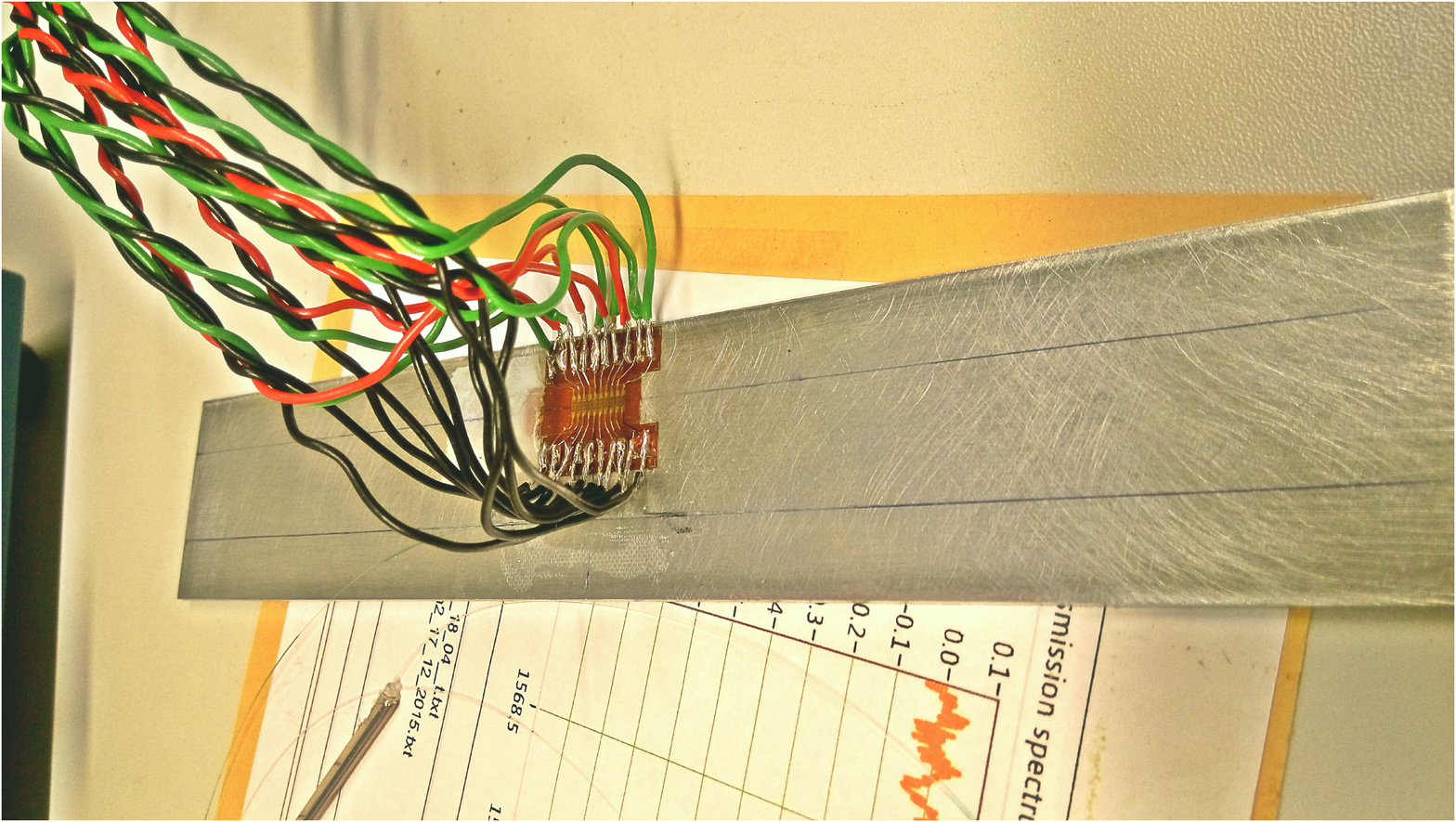}
  \caption{}
  \label{fig:sdist2}
\end{subfigure}
\begin{subfigure}{.5\textwidth}
  \centering
  \vspace{5mm}
  \includegraphics[width=.8\textwidth]{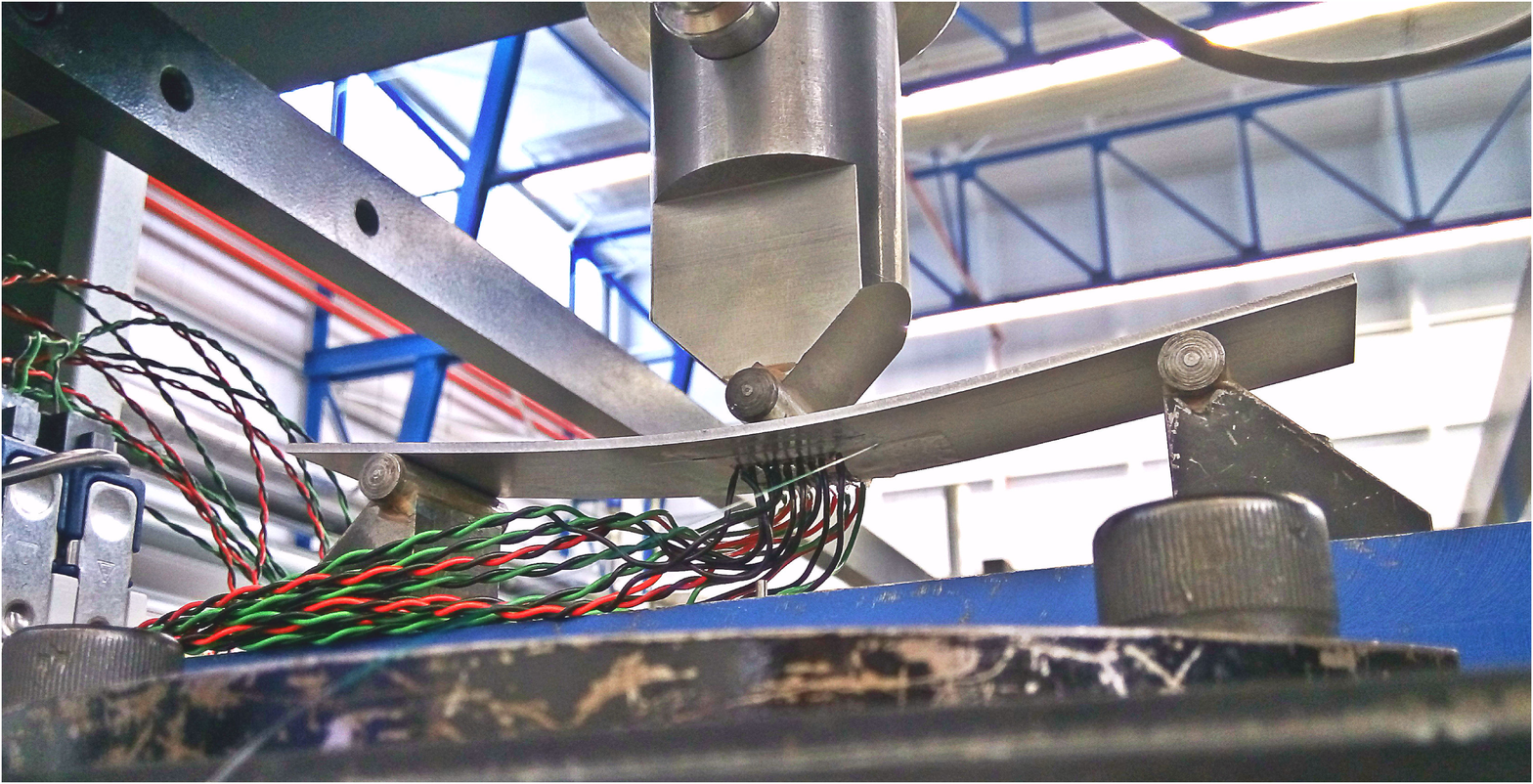}
  \caption{}
  \label{fig:sdist2}
\end{subfigure}
\begin{subfigure}{.5\textwidth}
  \centering
    \vspace{5mm}
  \includegraphics[width=.8\textwidth]{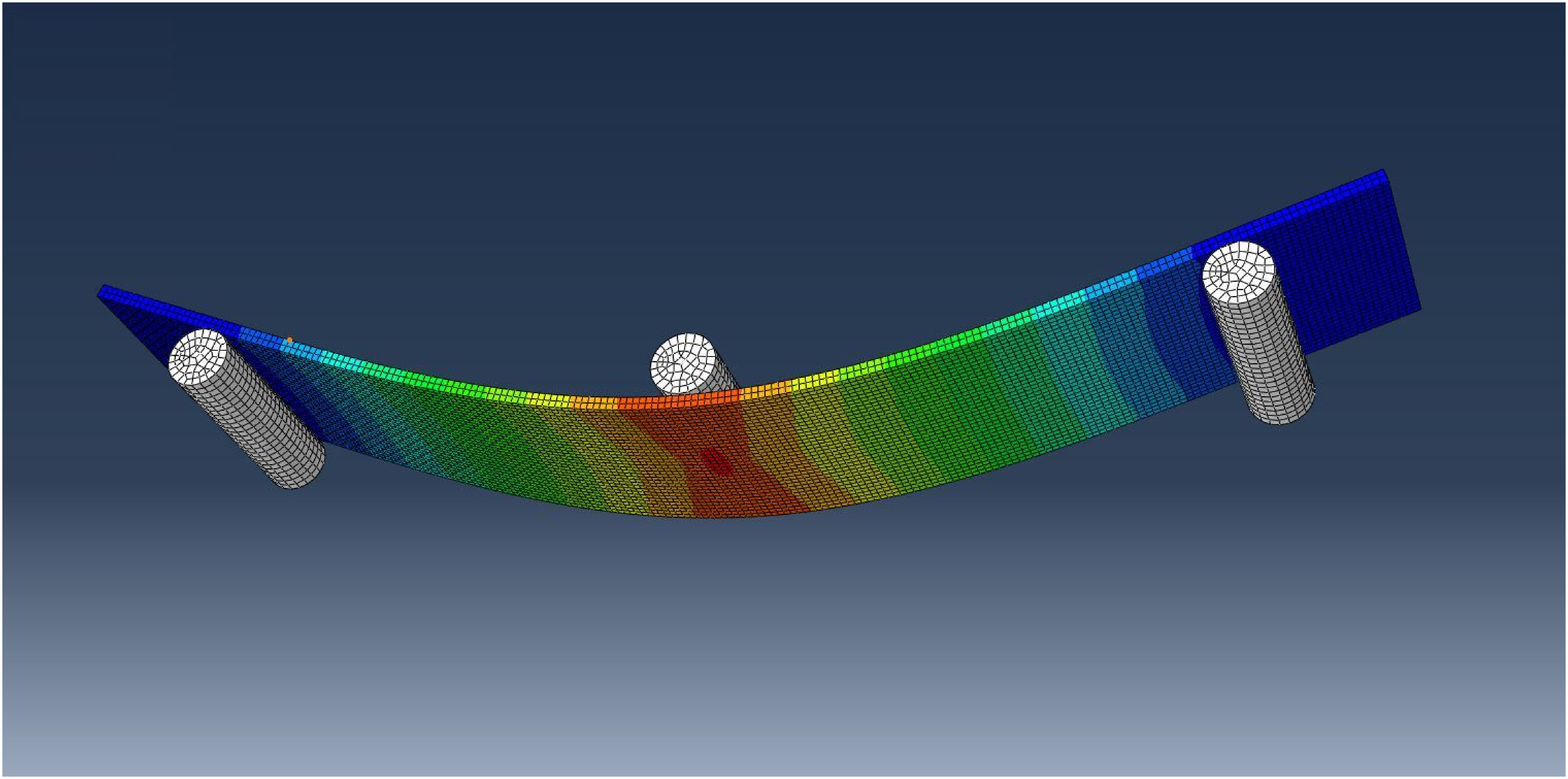}
  \caption{}
  \label{fig:sdist2}
\end{subfigure}
\caption{(a): A schematic design of the specimen and the location of the sensors. (b): The test specimen (bottom view) with the surface mounted strain gauge and the FBG. (c): The test specimen under load, while the signals were being recorded. (d): Finite element modelling of the specimen under a similar load. }
\label{fig:exp_setup}
\end{figure*}
\begin{figure*}[t]
\begin{subfigure}{.5\textwidth}
  \centering
  \includegraphics[width=0.97\textwidth]{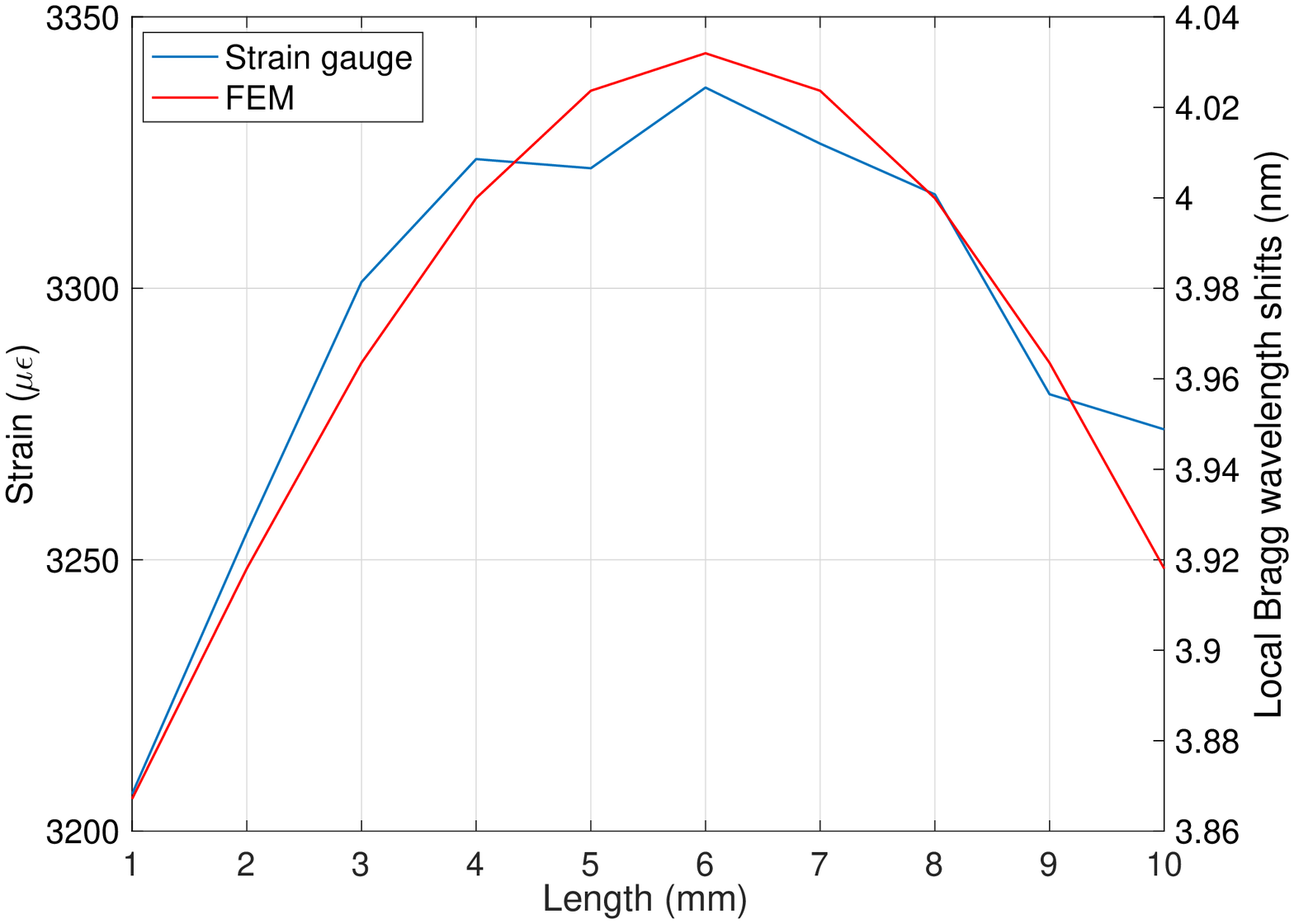}
  \caption{}
  \label{fig:gaugedist1}
\end{subfigure}%
\begin{subfigure}{.5\textwidth}
  \centering
  \includegraphics[width=0.97\textwidth]{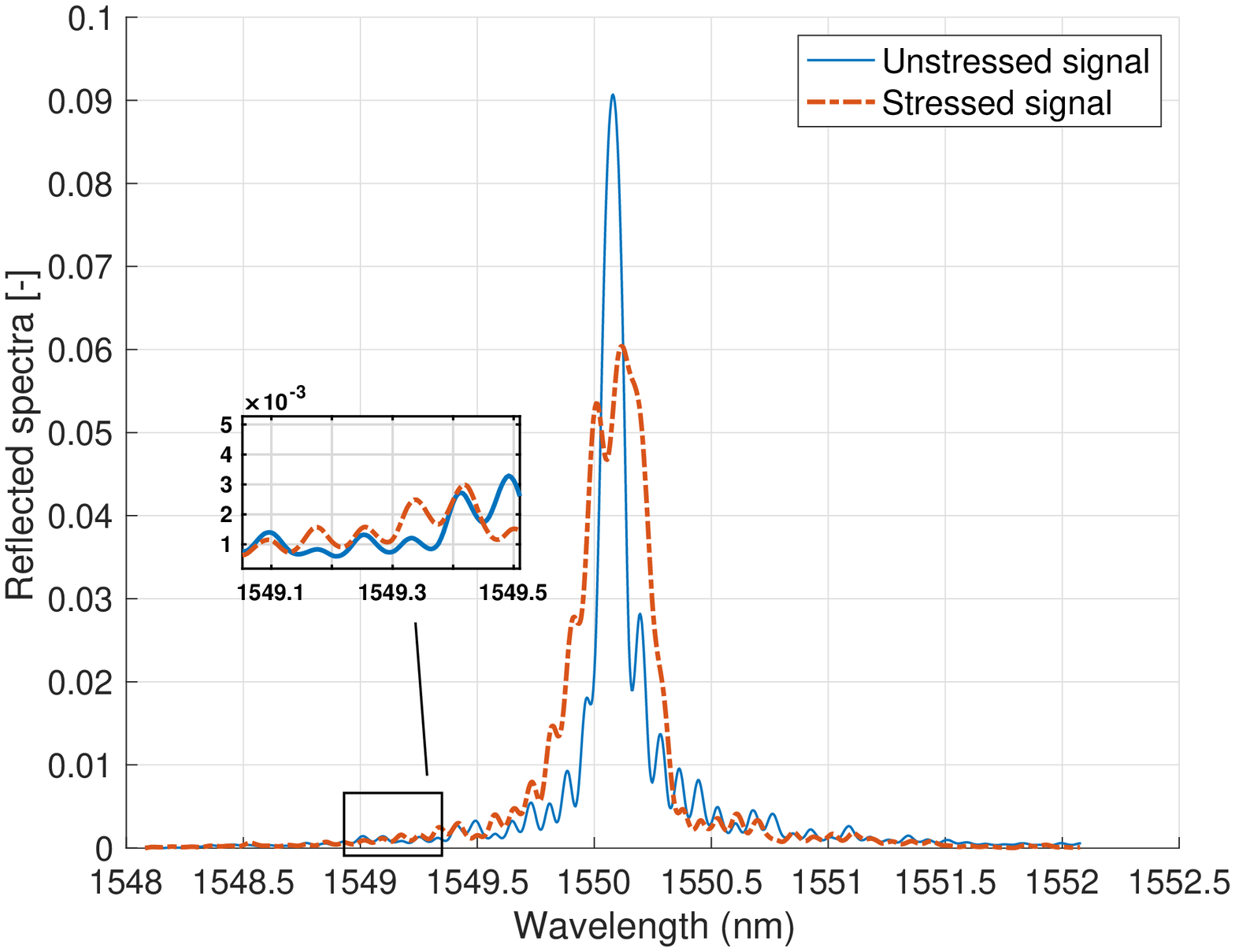}
  \caption{}
  \label{fig:gaugedist2}
\end{subfigure}
\caption{(a): The strain distributions recorded from the strain gauges (in blue) and calculated with FEM (in red). (b): The primary and secondary FBG reflected spectra are aligned with respect to maximising the correlation of their side peaks.}
\label{fig:exp_res}
\end{figure*}

 \color{black} Since the information in the side lobes of the reflected spectra plays a significant role in this study, a high dynamic range interrogator based on a fibre Fabry-Perot tunable filter (FFP-TF) from Micro Optics was used to record the output of the sensors. The PXIe-4844 FBG interrogator from National Instruments, which has a dynamic range of $40 \, \rm{dB}$ and a wavelength accuracy of $4 \, \rm{pm}$, was used to interrogate the sensors. The wavelength range for this device ranges from $1510 \, \rm{nm}$ to $1590 \, \rm{nm}$. 

In order to validate the results, we used both electrical strain gauge measurements and results obtained by finite element modelling (FEM) as reference strain measurements. For the strain gauge measurements, an array of $10$ miniature strain gauges with a pitch of $1 \, \rm{mm}$ (HBM 1-KY11-1/120), spatial resolution of $1\,\rm{mm}$, and nominal strain accuracy of less than $5 \, \mu \varepsilon$ was used. The data acquisition of the analogue output of the strain gauges was performed using the NI-9219 universal analogue input modules from National Instruments.
The data acquisition was carried out using the National Instruments LabVIEW software, and the signal processing and conditioning was performed in MATLAB R2016b. 
To avoid unwanted artefacts and complications associated with finite element modelling, the test setup was designed as simply as possible by using specimens with isotropic properties and subjecting them to a static three point flexural test. The FBG sensor was surface mounted to a piece of 6082 aluminium alloy with dimensions $200 \,\rm{mm} \times 35 \, \rm{mm} \times 2 \, \rm{mm}$, and the strain gauge array was pasted on the specimen in a symmetrical position with respect to the FBG. Fig.~\ref{fig:exp_setup}a and \ref{fig:exp_setup}b depict the configuration of this setup. To induce the desired strain field, the specimen was placed in a $10 \, \rm{kN}$ tensile machine, with a force of $150 \, \rm{N}$ applied on the loading pin (see Fig.~\ref{fig:exp_setup}c). The finite element modelling  of the specimen under such force is shown in Fig.~\ref{fig:exp_setup}d in which the element size was set to $1\, \rm{mm}$, equal to the spatial resolution of the strain gauges. Note that both sensors are adhered to the specimen on the opposite to the contact surface of the loading pin.

Based on the FEM analysis and the recordings of the strain gauges, the strain distributions shown in Fig.~\ref{fig:exp_res}a were obtained. The mean strain value of the FEM was $3291 \, \mu \varepsilon$, whereas the recordings of the electrical strain gauges showed a mean strain value of $3294 \, \mu \varepsilon$. 
In order to estimate the mean strain with the proposed method, we again follow the steps outlined in  Algorithm~\ref{alg:strain}. 
That is, we first align the centre of masses of the reflected spectra of the unstrained and strained sensor, resulting in a wavelength shift of $\Delta\lambda_{B_c} = 3.932 \, {\rm nm}$, after which we maximise the cross-correlation of the side lobes, as shown in Fig.~\ref{fig:gaugedist2}, resulting in an additional wavelength shift of $\delta\lambda_B =+56 \, {\rm pm}$. With this, the resulting mean strain estimate becomes  $\bar{s} = (3932+56)/k_s=3298 \, \mu \varepsilon$. Note that classical methods would estimate the mean strain value from \eqref{eq:ks} resulting in $\bar{s} = \Delta\lambda_B/k_s = 3328  \, \mu \varepsilon$. Furthermore, we applied different force loads on the loading pin of the same experimental setup, in which we started from $10 \,\text{N}$ and then linearly increased the force to $150 \, \text{N}$ with increments of $10 \,\text{N}$. The results of the compensated mean strain estimation using our algorithm, and the mean strain values based on the shift of the Bragg wavelength are presented in Fig.~\ref{fig:more_exp}, and the compensated mean strain values are in accordance with the data from the electrical strain gauges. It is noteworthy that as the applied force increases, the error associated with the traditional strain estimation method increases too, which is due to further deviating from a uniform strain field. 
 
\begin{figure*}[t]
\begin{subfigure}{.5\textwidth}
  \centering
  \includegraphics[width=0.97\textwidth]{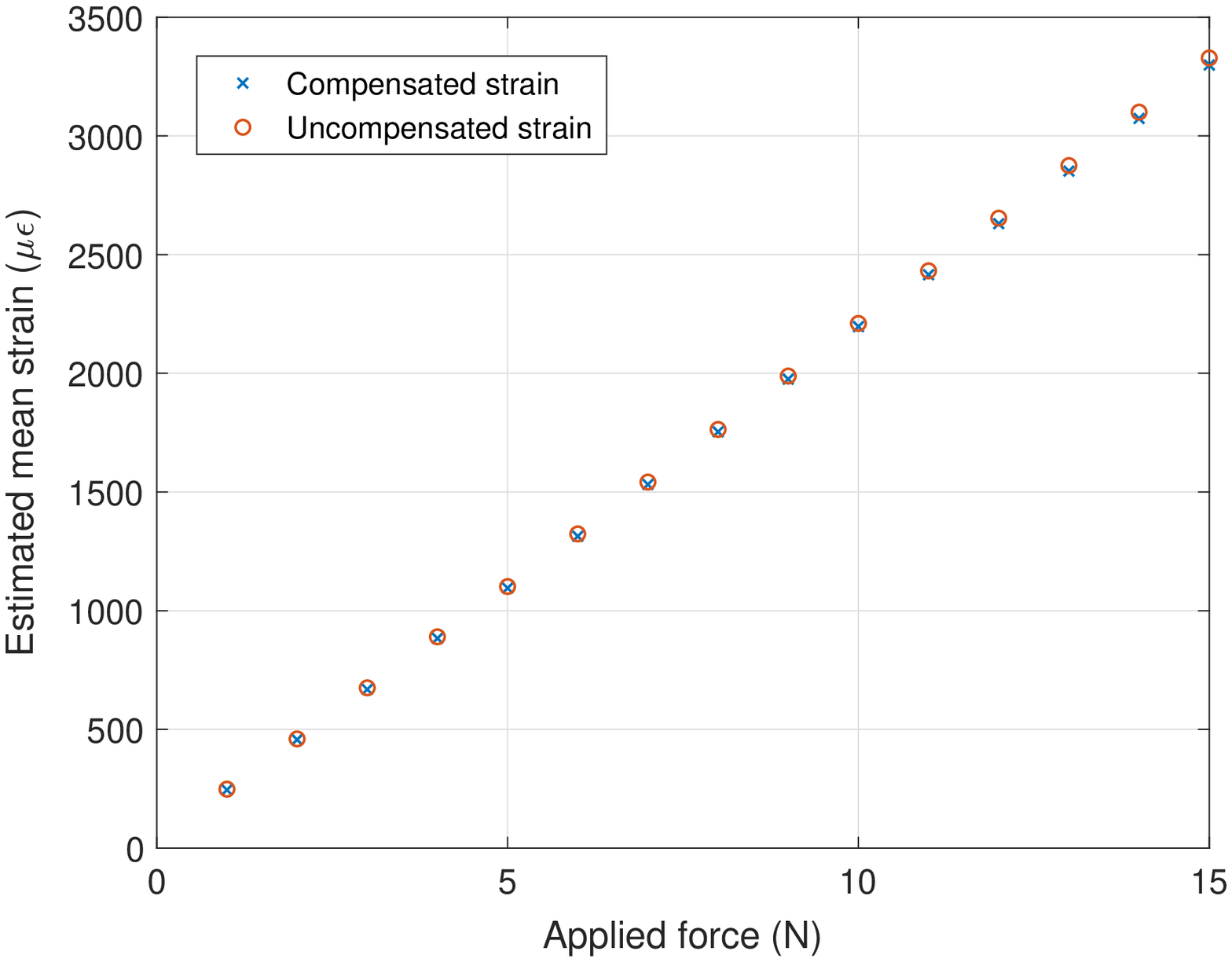}
  \caption{}
  \label{fig:more1}
\end{subfigure}%
\begin{subfigure}{.5\textwidth}
  \centering
  \includegraphics[width=0.97\textwidth]{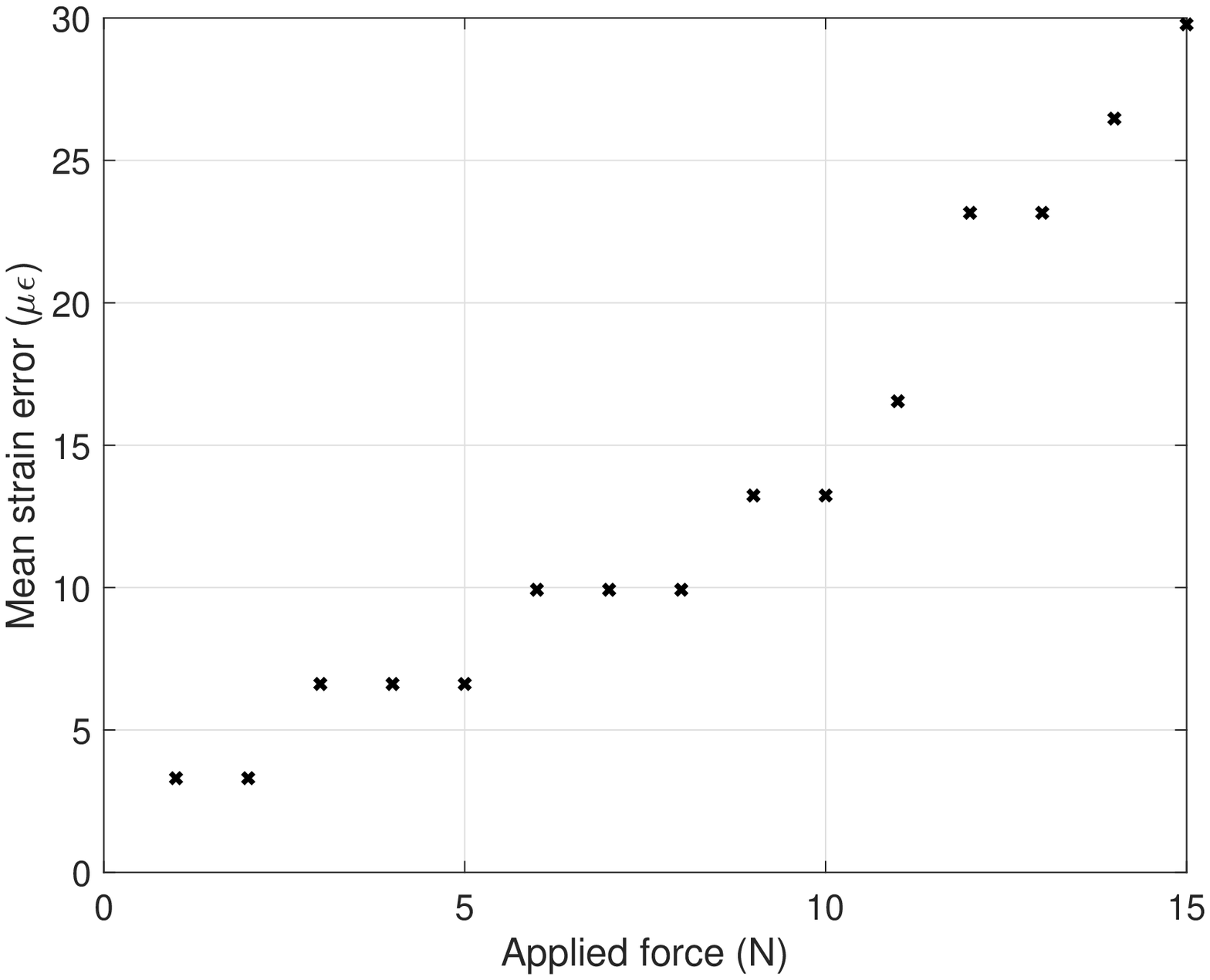}
  \caption{}
  \label{fig:more2}
\end{subfigure}
\caption{(a): Strain estimation using the traditional strain measurement method (red), and the compensated mean strain values based on our algorithm (blue). (b): The mean strain error associated with uncompensated mean strain estimations.}
\label{fig:more_exp}
\end{figure*}

\color{black}
\section {Conclusions} \label{conclusions}
In this paper, we focussed on estimating the mean strain value in the case of  smoothly varying non-uniform strain distributions. \color{black} We proposed a new algorithm that accurately estimates the mean strain value and showed that this shift is related to the average shift of the peak wavelength along the length of the sensor. In order to find this average shift, we introduced an approximation of the well known transfer matrix model and we showed that the information we need can be found by inspection of the side lobes of the reflected spectra. That is, the maximum likelihood estimator of the mean strain is obtained by cross-correlating the side lobes of the reflected spectra of the strained and unstrained sensor. In order to overcome possible estimation problems in practical scenarios, we proposed an alternative two step algorithm where we first measure the Bragg wavelength shift as is done in traditional strain estimation algorithms, and then refine the estimate by cross-correlating the side lobes of both spectra over a small range around the shifted Bragg wavelength. We validated the algorithm using both computer simulations and experimental FBG measurements and showed that the newly proposed algorithm clearly outperforms state-of-the-art strain estimation algorithms by compensating for mean strain errors of around $60 \mu\varepsilon$.  However, in case of non-smooth strain distributions with high variations, and also under extreme birefringence effects, the cross correlation function could lead to local maxima and an incorrect mean Bragg wavelength retrieval. Developing a more robust technique for retrieving the mean Bragg wavelength should be the focus of future studies.

\color{black}

\appendices
\section{Proof of Lemma~\ref{lem:approx}}
\label{app:approx}
\begin{enumerate}
\item By definition, we have
\begin{align*}
\bar{\alpha} &= \frac{1}{M}\sum_{i=1}^M\alpha_i \\
&= 2\pi n_{\rm eff}\Delta z \;  \frac{1}{M}\sum_{i=1}^M \frac{1}{\lambda_{B_i}} \\
&= 2\pi n_{\rm eff}\Delta z \;  \frac{1}{M}\sum_{i=1}^M \frac{1}{\bar{\lambda}_B+\Delta_i}.
\end{align*}
Taylor series expansion of the terms in the summation yields
\begin{align*}
 \frac{1}{\bar{\lambda}_B+\Delta_i} &=  \frac{1}{\bar{\lambda}_B\left(1 +\frac{\Delta_i}{\bar{\lambda}_B}\right)} \\
 &= \frac{1}{\bar{\lambda}_B}\left( 1 - \frac{\Delta_i}{\bar{\lambda}_B} + 
 {\cal O}\left(\!\left(\frac{\Delta_i}{\bar{\lambda}_B}\right)^2\right)\!\right),
\end{align*}
so that
\[
\bar{\alpha} = \frac{2\pi n_{\rm eff}\Delta z}{\bar{\lambda}_B} +   {\cal O}\left(\frac{\overline{\Delta^2}}{\bar{\lambda}_B^3}\right), 
\]
since 
\[
  \frac{1}{M}\sum_{i=1}^M\Delta_i =   \frac{1}{M}\sum_{i=1}^M (\lambda_{B_i} - \bar{\lambda}_B) = 0.
\]
\item
Let $k \in \{1,M\}$. We have
\[
\frac{\xi_i-\xi_{i+1}}{\xi_k} = \frac{(\alpha-\alpha_i)^{-1} - (\alpha-\alpha_{i+1})^{-1}}{(\alpha-\alpha_k)^{-1}},
\]
for $i=1,\ldots,M-1$.
Since
\[
(\alpha-\alpha_i)^{-1} = \frac{\lambda \lambda_i}{\lambda_i-\lambda},
\]
we conclude that
\[
\frac{\xi_i-\xi_{i+1}}{\xi_k} = \frac{(\lambda_k-\lambda)\lambda}{\lambda_k (\lambda_i-\lambda)(\lambda_{i+1}-\lambda)}(\lambda_{i+1}-\lambda_i).
\]
Let $\Delta_{\max} = \max_i|\Delta_i|$ and let $|\lambda-\bar\lambda_B| = \Delta\lambda > \lambda_{\rm th}$ so that $\Delta\lambda-\Delta_{\max} < |\lambda_i-\lambda| < \Delta\lambda+\Delta_{\max}$ for $i=1,\ldots,M$. With this, we have
\[
\left| \frac{\xi_i-\xi_{i+1}}{\xi_k}\right| \leq \frac{\Delta\lambda+\Delta_{\max}}{(\Delta\lambda-\Delta_{\max})^2}\frac{\bar\lambda_B + \Delta\lambda}{\bar\lambda_B-\Delta_{\max}}
|\lambda_{i+1}-\lambda_i|,
\]
so that $|\xi_i-\xi_{i+1}| \ll |\xi_k|$ if 
\[
|\lambda_{i+1}-\lambda_i| \ll \frac{(\Delta\lambda-\Delta_{\max})^2}{\Delta\lambda+\Delta_{\max}}\frac{\bar\lambda_B-\Delta_{\max}}{\bar\lambda_B + \Delta\lambda}< \Delta\lambda.
\]
Since $\lambda_{\rm th}$ is the infimum of $\Delta\lambda$, the condition
$|\lambda_{i+1}-\lambda_i| \ll \lambda_{\rm th}$ is a sufficient condition for $|\xi_i-\xi_{i+1}| \ll |\xi_k|$.
This completes the proof.
\end{enumerate}

\section*{Acknowledgement}
This research is part of the TKI Smart Sensing for Aviation Project, sponsored by the Dutch Ministry of Economic Affairs under the Topsectoren policy for High Tech Systems and Materials, and industry partners Airbus Defence and Space, Fokker Technologies | GKN Aerospace and Royal Schiphol Group.

\bibliography{Aydinlib} 
%

\end{document}